\documentclass[journal]{IEEEtran}

\usepackage{amsmath,amsfonts,amssymb,amscd,amsthm,dsfont}
\usepackage{graphicx,epsfig}
\usepackage{color}
\usepackage{mathtools}
\usepackage{bbm}
\usepackage{url}
\usepackage{balance}
\usepackage{epstopdf}
\usepackage{subfigure,enumitem}
\usepackage{cite,url,hyperref}
\usepackage{verbatim}
\usepackage{bm}
\usepackage{multicol}
\usepackage{ulem}
\usepackage{booktabs}
\usepackage{multirow}

\def\ms{{\mathcal S_{n,\mathbf{\lambda}}}}

\newcommand{\mP}{\mathbb{P}}
\newcommand{\mE}{\mathbb{E}}

\newcommand{\nn}{\nonumber}

\newtheorem{theorem}{\textbf{Theorem}}

\newtheorem{lemma}{\textbf{Lemma}}
\newtheorem{proposition}{\textbf{Proposition}}

\usepackage{color}

\begin{document}
%
\title{Quickest Change Detection in Anonymous Heterogeneous Sensor Networks}
\author{Zhongchang Sun \quad Shaofeng Zou \quad Ruizhi Zhang  \quad  Qunwei Li 
	\thanks{This paper was presented in part at the 2020 IEEE International Conference on Acoustics, Speech and Signal Processing \cite{sun2020icassp}.}\thanks{ The work of Z. Sun and S. Zou was supported in part by NSF grants CCF-1948165 and ECCS-2112693. The work of R. Zhang was supported in part by NSF grant ECCS-2112740.}
	\thanks{Zhongchang Sun and Shaofeng Zou are with the Department of Electrical Engineering, University at Buffalo, Buffalo, NY 14228 USA (e-mail: \href{mailto:zhongcha@buffalo.edu}{zhongcha@buffalo.edu}, \href{mailto:szou3@buffalo.edu}{szou3@buffalo.edu}). Ruizhi Zhang is with the Department of Statistics, University of Nebraska-Lincoln, Lincoln, NE 68588 USA (e-mail: \href{mailto:rzhang35@unl.edu}{rzhang35@unl.edu}). Qunwei Li is with the Ant Group, Hangzhou, China (e-mail: \href{mailto:qunwei.qw@antfin.com}{qunwei.qw@antfin.com})}
}
\maketitle
\begin{abstract}
The problem of quickest change detection (QCD) in anonymous heterogeneous sensor networks is studied. There are $n$ heterogeneous sensors and a fusion center. The sensors are clustered into $K$ groups, and different groups follow different data-generating distributions. At some unknown time, an event occurs in the network and changes the data-generating distribution of the sensors. The goal is to detect the change as quickly as possible, subject to false alarm constraints. The anonymous setting is studied, where at each time step, the fusion center receives $n$ unordered samples, and the fusion center does not know which sensor each sample comes from, and thus does not know its exact distribution. A simple optimality proof is first derived for the mixture likelihood ratio test, which was constructed  and proved to be optimal for the non-sequential anonymous setting in \cite{anonymous}. For the QCD problem, a mixture CuSum algorithm is further constructed, and is further shown to be optimal under Lorden's criterion. For large networks, a computationally efficient test is proposed and a novel theoretical characterization of its false alarm rate is developed. Numerical results are provided to validate the theoretical results.
\end{abstract}


\section{Introduction}
\IEEEPARstart{I}n quickest change detection (QCD) problem \cite{veeravalli2013quickest,tartakovsky2014sequential,poor-hadj-qcd-book-2009,detectAbruptChange93,Siegmund1985,tartakovsky2019sequential,xie2021sequential}, a decision maker collects samples sequentially from a stochastic environment. At some unknown time, an event occurs and causes a change in the data-generating distribution. The goal of the decision maker is to detect the change as quickly as possible subject to a constraint on the false alarm. The QCD problem in sensor networks has been widely studied in the literature \cite{tartakovsky2004change,tartakovsky2006novel,mei2010efficient,xie2013sequential,fellouris2016second,raghavan2010quickest, he2006nonparametric, dechade2018, hadjiliadis2009one,ludkovski2012bayesian,zou2020dynamic,veeravalli2001decentralized,tartakovsky2008distributed,zou2019distributed}. In these studies,  it is usually assumed that the fusion center knows which sensor that each sample comes from, and thus the statistical property of the sample is known.
However, in a wide range of modern practical applications,  the nodes are anonymous and heterogeneous. In this case, only unordered and anonymous samples are available to the fusion center, and the fusion center doesn't know what data generating distribution that each sample follows.

In this paper, we investigate the QCD problem using anonymized samples. We consider a general scenario with heterogeneous sensors, where the sensors can be clustered into $K$ groups with different data generating distributions, and the fusion center does not know which sensor each sample comes from. At some unknown time, an event occurs in the network, and changes the data-generating distribution of the nodes. The goal is to detect the change as quickly as possible subject to false alarm constraints using anonymized samples (see Fig.~\ref{fig:model}). 

\begin{figure}[!htb]
	\centering
	\includegraphics[width=0.73\linewidth]{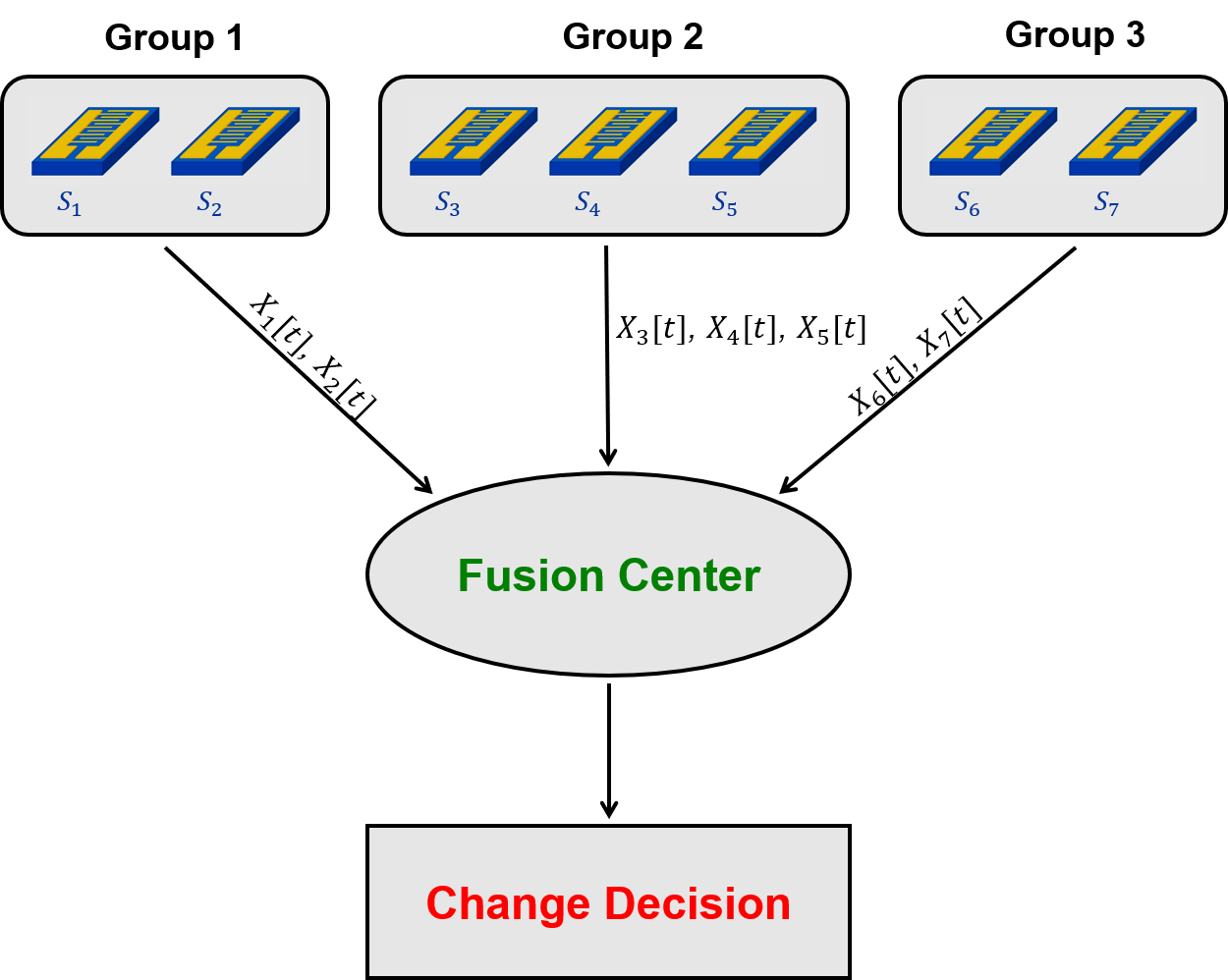}
	\caption{QCD in anonymous heterogeneous sensor network. Sensors in different groups have different distributions. At each time $t$, the fusion center collects unordered samples $X_1[t],\cdots, X_7[t]$ from sensors $S_1,\cdots,S_7$. }
	\label{fig:model}
\end{figure}

Statistical inference for anonymous and heterogeneous sensor networks finds a wide range of practical applications. For example, in large-scale Internet of things (IoT) networks \cite{luigi2010internet, rao2012cloud, li2013compressed, fuqaha2015iot}, devices are commonly small and low-cast sensing devices powered by battery, and are usually deployed in a massive scale. In such large-scale networks, the communication cost of identifying individual sensors increases drastically as the number of sensors grows \cite{anonymous}, which is not affordable for battery powered small IoT sensing devices that are expected to survive for years without battery change. Moreover, sensors in IoT networks are usually heterogeneous for various industrial and consumer applications, e.g., pressure sensor, light sensor, temperature sensor, humidity sensor, seismic sensor and electrochemical sensor. The same type of sensors deployed in different regions also exhibits heterogeneity in their data generating distributions, e.g., electrochemical sensors that are near to or far away from the air pollution source  and climate sensors on different sides of the same mountain. 
The second example is crowdsourcing, which is an evolving distributed problem-solving and business production model \cite{yuen2011crowd, enrique2012towards, brabham2013crowdsourcing}. Crowdsourcing aims to collect data, ideas, micro-tasks from a large and relatively open group of people. With human participants, anonymity is necessarily needed to protect privacy \cite{zhou2008brief, cummings2018differentially,dwork2008differential,dwork2014algorithmic, horey2007anonymous,groat2011kipda,alomair2010statistical,Wadaa2004}. Based on their skill level and background, e.g., education, country, and language, participants can be divided into groups that are heterogeneous. QCD finds a wide range of applications in these networks, e.g., environmental change (air/water quality) monitoring,  fake news detection in social networks, pandemic outbreak detection and seismic wave detection. In these applications, a change in the data-generating distributions occurs due to an abrupt event which is of interest to be detected quickly.

\subsection{Related Works}
The problem studied in this paper is closely related to the problem of QCD under the multiple-channel setup \cite{tartakovsky2004change,tartakovsky2006novel,mei2010efficient,xie2013sequential,fellouris2016second,raghavan2010quickest, hadjiliadis2009one,ludkovski2012bayesian,zou2020dynamic,veeravalli2001decentralized,tartakovsky2008distributed,zou2019distributed}, where samples are collected from multiple sensors sequentially, and the goal is to detect a change in the data-generating distribution of some unknown subset of the sensors or all the sensors. These works assume that the sensors are non-anonymous, i.e., it is known that which sensor that each sample comes from. In the non-anonymous setting, algorithms can be designed by combining the CuSum statistics each calculated for one sensor. These algorithms inherit  the nice property of the CuSum algorithm \cite{page1954continuous} which can be updated in an online and recursive fashion, and thus is computationally efficient. It was shown in these studies that such a type of algorithms are asymptotically optimal for various scenarios.
In this paper, we are interested in the anonymous setting, where at each time step the fusion center receives $n$ arbitrarily permuted (unordered) observations, and the permutations at different time steps may be different due to anonymity. Then, the fusion center does not know which samples over time come from one particular sensor. Therefore, existing approaches based on the idea of combining local CuSum statistics are not applicable any more since the fusion center is not able to compute one CuSum statistic for each node. 

In anonymous networks, the fusion center does not know the exact distribution of each sample due to the uncertainty caused by the anonymity  \cite{anonymous}. The group label that assigns the samples to different types of sensors is an unknown parameter of the distribution. Therefore, the QCD problem in anonymous networks can be viewed as a composite QCD problem with unknown pre- and post-change distributions \cite{siegmund1995using,lai1998information,banerjee2015composite,zou2018quickest,mei2006sequential,Brodsky2005,pergamenchtchikov2019asymptotically}, 
The main difference lies in that the unknown parameter in our problem is changing with time, i.e., the group label may not be the same at different time steps, and thus the samples are not identically distributed in the pre- or post-change regime. As will be shown in our numerical results, a generalized likelihood ratio based test does not work well here. Furthermore, we do not assume that the distributions belong to any parameterized family of distributions, e.g., exponential family.

The problem of quickest detection of a moving anomaly was studied in \cite{rovatsos2019dynamic,rovatsos2020quickest}, where an unknown sensor is affected by an anomaly  with an unknown trajectory that emerges in the network at some unknown time.  In \cite{rovatsos2019dynamic}, the statistical behavior of the samples is modeled using a hidden Markov model \cite{fuh2018asymptotic}, and the trajectory is modeled as a deterministic and unknown one in \cite{rovatsos2020quickest}. Our work is different from the one in \cite{rovatsos2019dynamic} since we do not put any assumption on the prior of group label (trajectory of the anomaly in \cite{rovatsos2019dynamic}). The study in \cite{rovatsos2020quickest} is related to ours in that the samples can be equivalently viewed as being collected from anonymous sensors since the node affected by the anomaly is unknown.

The offline hypothesis testing problem in the anonymous setting was investigated in \cite{anonymous}, where one sample is collected from each sensor. A mixture likelihood ratio test (MLRT) was  developed, and was further shown to be optimal under the Neyman-Pearson setting. Here, we consider the QCD problem under the anonymous setting with sequential samples and time-varying group labels. We are interested in the tradeoff between the false alarm rate and average detection delay, which requires construction of sequential tests and more involved optimality analysis. 

In Table \ref{table:summary}, we summarize the difference between our paper and other related works. We note that the fusion center may be able to recover the group identity if it performs, e.g.,  hypothesis testing, and the error probability depends on the distance between the distributions of different groups. With unordered samples, perfect anonymity can only be guaranteed if distributions among different groups are exactly the same. Designing optimal mechanisms to achieve perfect anonymity is not the focus of this paper, and might be of independent interest. In this paper, we focus on the design of optimal quickest detection algorithms for the scenario with unordered samples.

\begin{table*}[!htbp]
	\centering	
	\begin{tabular}{|c|c|}
		
		\hline
		Existing studies&Our paper\\\cline{1-2}
		\cite{tartakovsky2004change,tartakovsky2006novel,mei2010efficient,xie2013sequential,fellouris2016second,raghavan2010quickest, hadjiliadis2009one,ludkovski2012bayesian,zou2020dynamic,veeravalli2001decentralized,tartakovsky2008distributed,zou2019distributed}: QCD problem in sensor networks, sensors are & Sensors are anonymous, i.e., fusion center does not know\\
		non-anonymous &which sensor each sample comes from\\\cline{1-2}
		
		\cite{siegmund1995using,lai1998information,banerjee2015composite,zou2018quickest,mei2006sequential,Brodsky2005,pergamenchtchikov2019asymptotically}:Composite QCD problem with unknown pre- and &Unknown parameters change with time, and generalized likelihood \\ 
		post-change distributions &ratio based test is not applicable\\
		\cline{1-2}
		
		\cite{rovatsos2019dynamic, rovatsos2020quickest}: quickest detection of moving anomaly & Samples are unordered at each time step\\ \cline{1-2}
		
		\cite{anonymous}: offline binary hypothesis testing problem in anonymous networks & Samples are sequentially collected\\
		\hline
	\end{tabular}
	\vspace{0.1cm}\caption{Comparison of related works}\label{table:summary}
\end{table*}
\subsection{Main Contributions}
We first revisit the non-sequential hypothesis testing problem with anonymous heterogeneous sensors. We provide a simple proof for the optimality of the MLRT \cite{anonymous}. The basic idea is to construct a binary composite hypothesis testing problem with uniform priors on all possible group labels, and to show that the optimal test for the case with a uniform Bayesian prior is also optimal under the minimax setting.

For the QCD problem in anonymous networks, we design a mixture CuSum algorithm, and prove that the mixture CuSum algorithm is exactly optimal under Lorden's criterion \cite{lorden1971procedures}. To show its exact optimality, we build a novel connection among several simple QCD problems and the QCD problem under the anonymous setting. The major challenge in our analysis is due to that we are optimizing the worst-case performance over all possible change-point, group labels and pre-change observations. 

The computational complexity of the mixture CuSum algorithm at each time step increases almost exponentially in the number of nodes, and thus is not efficient when the network is large. We then propose a computationally efficient test based on the asymptotic behavior of the mixture CuSum test statistic when the network is large. The basic idea is to approximate the mixture CuSum statistic by a convex optimization problem with linear constraints, the computational complexity of which is independent of the number of sensors. We provide a comprehensive discussion of its performance. We also derive a lower bound on its worst-case average run length to false alarm, so that a threshold can be chosen analytically for false alarm control in practice.


We provide numerical results to demonstrate the performance of our proposed algorithms. We compare our algorithms with two other heuristic algorithms based on the Bayesian approach and the generalized likelihood ratio approach, and show that our mixture CuSum has the best performance, and our computationally efficient test also performs better than the other two tests. We also compare their computational complexity, and show that our the computationally efficient algorithm is much more efficient for large networks.

\subsection{Paper Organization}
 In Section \ref{sec:2problemmodel}, we present the problem formulation. In Section \ref{sec:3anonbinHT}, we revisit the binary hypothesis testing problem in the anonymous setting, and derive a simple optimality proof for the MLRT. In Section \ref{sec:4twoalogrithm}, we develop the mixture CuSum algorithm, and prove its exact optimality under Lorden's criterion. We further develop a computationally efficient test and characterize its performance theoretically. In Section \ref{sec:7numerical}, we provide numerical results to validate our theoretical assertions. In Section \ref{sec:8conclusion}, we present some concluding remarks.

\section{Problem Formulation}\label{sec:2problemmodel}
Consider a network consisting of $n$ sensors. The sensors are heterogeneous and can be divided into $K$ groups. Each group $k$ has $n_k$ sensors, $1\leq k\leq K$. 
The distributions of the observations in group $k$ are $p_{\theta,k}$,  $\theta\in \{0,1\}$. 
Let $\bm P_\theta = [p_{\theta,1}\cdots p_{\theta,K}]^T$.
We assume that $\bm\alpha^T\bm P_0 \neq \bm\alpha^T\bm P_1$. 
The centralized setting is considered, where there is a fusion center.
The sensors are anonymous, i.e., the fusion center does not know which group of sensors that each observation comes from. 
The fusion center only knows the distributions $p_{\theta,k}$, $\theta\in \{0,1\}$ and the number of sensors $n_k$ in each group $k$. 

\subsection{Binary Composite Hypothesis Testing}\label{sec:modelHT}
We first revisit the binary hypothesis testing problem in \cite{anonymous}. The goal is to distinguish between the two hypotheses: $\mathcal H_0: \theta=0$ and $\mathcal H_1: \theta=1$. 

Denote by $X^n=\{X_1,\ldots,X_n\}$ the $n$ collected samples. 
Denote by $\sigma(i)\in\{1,\ldots,K\}$ the label of the group that $X_i$ comes from, i.e., $X_i\sim p_{\theta, \sigma(i)}$. Due to the anonymity, $\sigma(i)$, $i=1,\ldots,n$, are \textit{unknown} to the fusion center.  There are $\left(\substack{n\\n_1,\ldots,n_K}\right)$ possible $\sigma: \{1,\ldots,n\}\rightarrow\{1,\ldots,K\}$ satisfying
$|\{i:\sigma(i)=k\}|=n_k,\forall k=1,\ldots,K.$
We denote the collection of all such labels by $\mathcal S_{n,\mathbf{\lambda}}$, where $\mathbf\lambda=\{n_1,\ldots,n_K\}$.

Given $\sigma$, the $n$ collected samples are assumed to be independent. The problem is a composite hypothesis testing problem, where $\sigma$ is the unknown parameter for both $\theta=0$ and $1$:
\begin{flalign}
\mathcal H_\theta: X^n\sim \mP_{\theta,\sigma}\overset{\Delta}{=}\prod_{i=1}^np_{\theta,\sigma(i)}, \text{ for some } \sigma\in\mathcal S_{n,\mathbf{\lambda}}.
\end{flalign}
The worst-case type-\uppercase\expandafter{\romannumeral1} and type-\uppercase\expandafter{\romannumeral2} error probabilities for a decision rule $\phi$ are defined as 
\begin{flalign}
P_F(\phi) &\triangleq \max\limits_{\sigma\in \ms}\mathbb{E}_{{0,\sigma}}[\phi(X^n)] ,\\
P_M(\phi) &\triangleq \max\limits_{\sigma\in \ms} \mathbb{E}_{{1,\sigma}}[1-\phi(X^n)],
\end{flalign}
where $\mathbb{E}_{{\theta,\sigma}}$ denotes the expectation under $\mathbb{P}_{{\theta,\sigma}}$, for $\theta \in \{0,1\}$ and $\sigma\in\mathcal S_{n,\mathbf{\lambda}}$.
The Neyman-Pearson setting is studied, where the goal is to solve the following problem for any $\zeta\in[0,1]$: 
\begin{flalign}\label{eq:problem}
\inf_{\phi:P_F(\phi) \leq \zeta} P_M(\phi).
\end{flalign}

\subsection{Quickest Change Detection}\label{sec:modelqcd}
In the QCD setting, anonymized samples are observed sequentially. At some unknown time $\nu$, an event occurs in the network, and changes the data-generating distributions of the sensors. Specifically, denote the $i$-th sample at time $t$ by $X^n_i[t]$ and all the observed samples at time $t$ by $X^n[t]$. Before the change, i.e., $t<\nu$, $X^n[t]\sim \mP_{0,\sigma_t}, \text{ for some \textit{unknown} } \sigma_t\in\mathcal S_{n,\mathbf{\lambda}}.$ After the change, i.e., $t\geq \nu$, $X^n[t]\sim \mP_{1,\sigma_t}, \text{ for some \textit{unknown} } \sigma_t\in\mathcal S_{n,\mathbf{\lambda}}.$
We note that $\sigma_t$ may change with time, i.e., $\sigma_{t_1}$ may not be the same as $\sigma_{t_2}$, for $t_1\neq t_2$. We assume that for any $t\geq 0$, given $\sigma_t$, the samples in $X^n[t]$ are independent. We further assume that $X^n[{t_1}]$ is independent from $X^n[{t_2}]$ for any $t_1\neq t_2$.

The objective is to detect the change at time $\nu$ as quickly as possible subject to false alarm constraints. 
In this paper, we consider a deterministic unknown change point $\nu$. We define the worst-case average detection delay (WADD) under Lorden's criterion \cite{lorden1971procedures} and worst-case average run length (WARL) for any stopping time $\tau$ as follows:
\begin{flalign}\label{eq:5}
&\text{WADD}(\tau) \triangleq\sup_{\nu\geq 1}\sup_{\Omega}\text{ess}\sup \mathbb{E}^\nu_{\Omega}\left[(\tau-\nu)^+|\mathbf X^n[1,\nu-1]\right],\nn\\
&\text{WARL}(\tau)\triangleq \inf_{\Omega}{\mathbb{E}^\infty_{\Omega}[\tau]},
\end{flalign}
where $\Omega = \{\sigma_1, \sigma_2, ..., \sigma_\infty\}$, $\mathbb{E}^\nu_{\Omega}$ denotes the expectation when the change is at $\nu$, and the observations at time $t$ are labeled according to $\sigma_t$, and $\mathbf X^n[1,\nu-1]=\{X^n[1],\ldots,X^n[\nu-1]\}$.

The goal is to design a stopping rule that minimizes the WADD subject to a constraint on the WARL:
\begin{flalign}\label{eq:qcd}
\inf_{\tau:\text{WARL}(\tau)\geq \gamma} \text{WADD}(\tau).
\end{flalign}

\subsection{Notations}
In this section, we list the notations used in this paper.
\begin{itemize}
	\item $n$ denotes the number of sensors, $K$ denotes the number of groups and $n_k$ denotes the number of sensor in group $k$.
	\item $\bm{\alpha} = [\alpha_1\cdots\alpha_K]^T$, where $\alpha_k =\lim_{n\rightarrow \infty} \frac{n_k}{n}$ denotes the asymptotic fraction of sensors of group $k$.
	\item $\sigma(i)\in\{1,\ldots,K\}$ denotes the label of the group that $X_i$ comes from, i.e., $X_i\sim p_{\theta, \sigma(i)}$, and $\mathcal S_{n,\mathbf{\lambda}}$ denotes the collection of all $\sigma(i)$, where $\mathbf\lambda=\{n_1,\ldots,n_K\}$.
	\item $H(\bm\alpha)$ denotes the entropy of $\bm\alpha$.
	\item $\Pi_{X^n}$ denotes the empirical distribution of samples $X^n$, and $T(\Pi_{X^n})$ denotes the type class of $\Pi_{X^n}$.
	\item $\mathcal{P}_n$ denotes the set of types with denominator $n$.
	\item $D(P\big|\big|Q)$ denotes the Kullback-Leibler (KL) divergence between $P$ and $Q$.
	\item $f(x) \sim g(x)$ as $x\rightarrow x_0$ if $f(x) = g(x)(1+o(1))$ as $x\rightarrow x_0$.		
\end{itemize}

\section{MLRT and A Simple Optimality Proof}\label{sec:3anonbinHT}
\label{sec:format}
For the binary composite hypothesis testing problem in Section \ref{sec:modelHT}, Chen and Huang constructed a mixture likelihood ratio test (MLRT), and showed that the MLRT is optimal under the Neyman-Pearson setting in \eqref{eq:problem} \cite{anonymous}. In this section, we first briefly review the optimality proof in \cite{anonymous}, and then we present a simple version of the proof.

Define the mixture likelihood ratio $\ell(X^n)$ as follows:
\begin{flalign}
\ell(X^n)=\frac{\sum_{\sigma\in\ms}\mP_{1,\sigma}(X^n)}{\sum_{\sigma\in\ms}\mP_{0,\sigma}(X^n)}.
\end{flalign}
Then the MLRT  was defined in \cite{anonymous} as 
\begin{flalign}\label{eq:mlrt}
\phi^*(X^n)=\left\{\begin{array}{ll}
    1, &\text{ if } \ell(X^n)> \eta \\
    \beta, &\text{ if } \ell(X^n)= \eta \\
    0, &\text{ if } \ell(X^n)< \eta,
\end{array}\right.
\end{flalign}
where $\beta\in[0,1]$, $\eta$ is the threshold, and they are chosen to meet the false alarm constraint. 

\begin{lemma}\label{lemma:1}\cite[Thm. 3.1]{anonymous}
The MLRT $\phi^*$ is optimal for \eqref{eq:problem}. 
\end{lemma}

The key idea of the proof in \cite{anonymous} is to reduce the original composite hypothesis testing problem in Section \ref{sec:modelHT} into a simple one through the ordering map $\Pi(X^n)$, and then apply the Neyman-Pearson lemma. The ordering map $\Pi(X^n)$ of $X^n$ is defined as $\Pi(X^n)=(X_{i_1},X_{i_2},\ldots,X_{i_n})$, such that $X_{i_1}\geq X_{i_2}\geq\ldots\geq X_{i_n}$. 
In the proof, due to the introduction of the ordering map, a careful examination of the measurability  needs to be conducted.
The proof in \cite{anonymous} can be summarized by the following steps. 1) In the auxiliary space induced by the ordering mapping, the induced probability measure is independent of $\sigma$, and thus the corresponding problem in the auxiliary space is a simple hypothesis testing problem. 2) In the auxiliary space,  applying the Neyman-Pearson lemma, the optimal test is obtained. 
3) Any symmetric test in the original sample space is equivalent to a test in the auxiliary space in terms of type-I and type-II error probabilities, where a test $\phi$ is symmetric if $\phi(x^n)=\phi(\pi(x^n))$ for any $x^n$ and any permutation $\pi$. 
4) The optimal test in the auxiliary space is the MLRT and is symmetric, which means that among all symmetric tests, the MLRT is optimal.   5) For any test $\psi$, one can always symmetrize it and construct a symmetric test $\phi$, which is as good as $\psi$.  6) Then, the MLRT is optimal among all tests.

In the following, we present a simple proof for the optimality of the MLRT. Our proof does not need to use the ordering map, and is much simpler.
\begin{proof}
We consider a Bayesian setting with a uniform prior on $\sigma$ under both hypotheses, 
and define the  average type-\uppercase\expandafter{\romannumeral1} and type-\uppercase\expandafter{\romannumeral2} error probabilities  for any test $\phi$:

\begin{flalign}
\widetilde{P}_F(\phi) &\triangleq \frac{1}{\mid \ms\mid}\sum_{\sigma\in \ms} \mathbb{E}_{{0,\sigma}}[\phi(X^n)] ,\label{eq:7}\\
\widetilde{P}_M(\phi) &\triangleq \frac{1}{\mid \ms\mid}\sum_{\sigma\in \ms} \mathbb{E}_{{1,\sigma}}[1-\phi(X^n)]. \label{eq:8}
\end{flalign}
Then under the Bayesian setting, this problem reduces to the following simple binary hypothesis testing problem:
\begin{flalign}
\mathcal H_0: \frac{1}{\mid \ms\mid}\sum\limits_{\sigma\in \ms}\mathbb{P}_{0,\sigma},\label{eq:h0bayesin}\\
\mathcal H_1: \frac{1}{\mid \ms\mid}\sum\limits_{\sigma\in \ms}\mathbb{P}_{1,\sigma},\label{eq:h1bayesin}
\end{flalign}
for which the optimal test (the same as the MLRT) is the likelihood ratio test between \eqref{eq:h0bayesin} and \eqref{eq:h1bayesin} \cite{moulin2018statistical}.

It can be verified that for any permutation $\pi(X^n)=(X_{\pi(1)},X_{\pi(2)},\ldots,X_{\pi(n)})$, $\phi^*(X^n)=\phi^*(\pi(X^n))$. For any $\pi$, let $\sigma'=\sigma\circ\pi$, where ``$\circ$'' denotes the composition of two functions, i.e., $f\circ g(x) = f(g(x))$. Then $\mE_{\theta,\sigma}[\phi^*(\pi(X^n))]=\mE_{\theta,\sigma\circ\pi}[\phi^*(X^n)]=\mE_{\theta,\sigma'}[\phi^*(X^n)]$.  For any $\sigma'\in\ms$, a $\pi$ can be found so that $\sigma\circ\pi=\sigma'$. Thus, for any $\sigma,\sigma'\in\ms$ and $\theta=0,1$,
\begin{flalign}
\mE_{\theta,\sigma'}[\phi^*(X^n)]=\mE_{\theta,\sigma}[\phi^*(X^n)].
\end{flalign}
It then follows that
\begin{flalign}
P_F(\phi^*) &= \max\limits_{\sigma\in \ms}\mathbb{E}_{{0,\sigma}}[\phi^*(X^n)] = \mathbb{E}_{0,\sigma}[\phi^*(X^n)] \nn\\
&= \frac{1}{\mid \ms\mid}\sum\limits_{\sigma\in \ms} \mathbb{E}_{{0,\sigma}}[\phi^*(X^n)]  = \widetilde{P}_F(\phi^*).
\end{flalign}
Similarly, it can be shown that $P_M(\phi^*)=\widetilde{P}_M(\phi^*)$.

From \eqref{eq:7} and \eqref{eq:8}, it follows that for any test $\phi$, 
\begin{flalign}\label{eq:12}
\widetilde{P}_F(\phi)&\leq {P}_F(\phi),\nn\\
\widetilde{P}_M(\phi)&\leq {P}_M(\phi). 
\end{flalign}
 Since $\phi^*$ is optimal for the problem of minimizing $\widetilde{P}_M(\phi)$ subject to $\widetilde{P}_F(\phi)\leq \epsilon$, then $\phi^*$ is also optimal for problem of minimizing ${P}_M(\phi)$ subject to ${P}_F(\phi)\leq \epsilon$. 
\end{proof}

\section{Mixture CuSum Algorithm and A Computationally Efficient Test}\label{sec:4twoalogrithm}

\subsection{Mixture CuSum Algorithm}\label{sec:exopmixture}
Motivated by the fact that the MLRT is optimal for the  binary composite hypothesis testing problem, we construct the following mixture CuSum algorithm:
\vspace{-0.2cm}
\begin{equation}\label{eq:mcusum}
    \tau^*(b) = \inf\Big\{t: \max\limits_{1\leq j\leq t}\sum_{i=j}^t \log\ell(X^n[i])\geq b\Big\}.
\end{equation}
Let $W[t] = \max\limits_{1\leq j\leq t}\sum_{i=j}^t \log\ell(X^n[i])$. The test statistic $W[t]$ has the following recursion:
\begin{flalign}
W[t+1] = (W[t])^+ + \log\ell(X^n[t+1]), W_0 = 0.
\end{flalign} 
The following theorem shows that the mixture CuSum algorithm is exactly optimal under Lorden's criterion \cite{lorden1971procedures} in \eqref{eq:qcd}.

\begin{theorem}\label{theorem:1}
Consider the QCD problem in Section \ref{sec:modelqcd}, the mixture CuSum algorithm in \eqref{eq:mcusum} is exactly optimal under Lorden's criterion in \eqref{eq:qcd}.
\end{theorem}

\begin{proof}[Proof Sketch]
Consider a simple QCD problem with samples independent and identically  distributed (i.i.d.) according to the pre-change distribution  $\widetilde{\mP}_0=\frac{1}{\mid \ms\mid}\sum_{\sigma\in \ms}\mP_{0,\sigma}$ and the post-change distribution  $\widetilde{\mP}_1=\frac{1}{\mid \ms\mid}\sum_{\sigma\in \ms}\mP_{1,\sigma}$, respectively. For this pair of pre- and post-change distributions, define the $\widetilde{\text{WADD}}$ and $\widetilde{\text{ARL}}$ for any stopping rule $\tau$ as follows:
\begin{flalign}
\widetilde{\text{WADD}}(\tau)&=\sup\limits_{\nu\geq1}\text{ess}\sup\widetilde{\mE}^\nu[(\tau-\nu)^+|\widetilde{\mathbf X}^n[1,\nu-1]],\nn\\
\widetilde{\text{ARL}}(\tau)&=\widetilde{\mE}^\infty[\tau],
\end{flalign}
where $\widetilde{\mE}^\nu$ denotes the expectation when the change is at $\nu$, the pre- and post-change distributions are $\widetilde{\mP}_0$ and $\widetilde{\mP}_1$, and $\widetilde{\mathbf X}^n[t], 1\leq t\leq \nu-1$, are i.i.d.\ from $\widetilde{\mP}_0$. For this new problem, the goal is to solve 
\begin{flalign}\label{eq:newproblem}
\inf_{\tau:\widetilde{\text{ARL}}(\tau)\geq \gamma} \widetilde{\text{WADD}}(\tau)
\end{flalign}
for some prescribed $\gamma>0$.
	
It was shown that the CuSum algorithm is exactly optimal for the problem in \eqref{eq:newproblem} under Lorden's criterion in \cite{MoustakidesAS86}. Therefore, $\tau^*$ in \eqref{eq:mcusum} is exactly optimal for the QCD problem defined by pre- and post-change distributions $\widetilde{\mP}_0$ and $\widetilde{\mP}_1$.
	
Following similar ideas as ones in Section \ref{sec:3anonbinHT}, we can show that for any stopping time $\tau$,
\begin{flalign}\label{EQ:AAA}
\widetilde{\text{WADD}}(\tau)\leq {\text{WADD}}(\tau) \text{ and }\widetilde{\text{ARL}}(\tau)\geq {\text{WARL}}(\tau).
\end{flalign}
	
We will then show that $\tau^*$ achieves the equality in \eqref{EQ:AAA}, which will complete the proof.
Due to the fact that $\tau^*$ is symmetric, i.e., it is invariant to any permutation of $X^n[j]$, $\forall j=1,2,\ldots$. For any $\Omega$ and $\Omega'$, it follows that 
\begin{flalign}\label{EQ:BBB}
\text{ess}\sup\mathbb{E}&^\nu_{\Omega}[(\tau^*-\nu)^+|\mathbf X^n[1,\nu-1]]\nn\\&=\text{ess}\sup\mathbb{E}^\nu_{\Omega'}[(\tau^*-\nu)^+|\mathbf X^n[1,\nu-1]],\nn\\
\mathbb{E}^\infty_{\Omega}[\tau^*]&=\mathbb{E}^\infty_{\Omega'}[\tau^*].
\end{flalign}
To establish \eqref{EQ:AAA} and the optimality of $\tau^*$, the proof is more involved than the binary hypothesis testing case in Sec. \ref{sec:3anonbinHT} due to the ess$\sup$  and the conditional expectation. 
\end{proof}
The missing details of the proof can be found in Appendix \ref{proff} and Appendix \ref{sec:appachiev}. The asymptotic optimality under Pollak's formulation \cite{pollak1985} can also be derived similarly, and is ignored in this paper however due to space limitation.

The mixture likelihood ratio $\ell(X^n[i])$ needs to compute the average of the likelihood over all possible $\sigma \in \ms$. Note that the size of $\ms$ is $\binom{n}{n_1,\cdots ,n_K}$. From the  exponential bounds
on the size of a type class\cite{cover2006information}, we have that 
$
\frac{2^{nH\big([\frac{n_1}{n}\cdots\frac{n_K}{n}]\big)}}{(n+1)^{\mid \mathcal{X}\mid}} \leq \binom{n}{n_1,\cdots, n_K} \leq 2^{nH\big([\frac{n_1}{n}\cdots\frac{n_K}{n}]\big)},\nn
$
where $H\big([\frac{n_1}{n}\cdots\frac{n_K}{n}]\big)$ denotes the entropy of $[\frac{n_1}{n}\cdots\frac{n_K}{n}]$. As $n\rightarrow \infty$, we have that $\lim_{n\rightarrow\infty}H\big([\frac{n_1}{n}\cdots\frac{n_K}{n}]\big) = H(\bm\alpha)$. Therefore, the computational complexity of mixture CuSum increases almost exponentially with $n$, which limits its practical applications in large networks. This motivates the need for computationally efficient tests  for large networks. There are a wide range of applications in which the number of nodes is very large, e.g., IoT networks with thousands of sensors, smart grids with a large number of PMUs, crowdsourcing, and wireless sensor networks.

\subsection{A Computationally Efficient Algorithm}
In this section, we focus on discrete distributions, that is, the cardinality of $\mathcal{X}$ is finite, where $\mathcal{X}$ denotes the alphabet of the distributions $p_{\theta,k}, \forall \theta \in \{0,1\}, k \in \{1,2,\cdots, K\}$. We note that our mixture CuSum algorithm and its exact optimality result apply  to general distributions, which are not necessarily discrete.
Denote by $\mathcal{P}_{\mathcal{X}}$ the set of all distributions supported on $\mathcal{X}$. We propose a computationally efficient algorithm and then derive a lower bound on its WARL so that a threshold can be chosen analytically for false alarm control. 

We first introduce some useful results that motivate the design of our algorithm.
Let $\Pi_{X^n}$ denote the empirical distribution of samples $X^n$, and let $T(\Pi_{X^n})$ denote the type class of $\Pi_{X^n}$.
Then,  it can be shown that \footnote{See Lemma 4.1 in \cite{anonymous} for the proof.}
\begin{flalign}\label{EQ:TYPE}
	\frac{\sum_{\sigma \in \ms}\mP_{1,\sigma}(X^n)}{\sum_{\sigma \in \ms}\mP_{0,\sigma}(X^n)} = \frac{\mP_{1,\sigma}\big(T(\Pi_{X^n})\big)}{\mP_{0,\sigma}\big(T(\Pi_{X^n})\big)}.
\end{flalign} 

The right hand side of equation in \eqref{EQ:TYPE} is a function of the empirical distribution $\Pi_{X^n}$. Let $\mathcal{P}_n$ denote the set of types with denominator $n$. For $n\geq 1$, let $Q_n\in\mathcal{P}_n$ be a sequence of distributions and $\lim_{n\rightarrow\infty}Q_n = Q$. The computation of the mixture likelihood ratio in \eqref{EQ:TYPE} can be approximated by an optimization problem when $n$ is large using the fact \footnote{See Lemma 5.2 in \cite{anonymous} for the proof.} that
\begin{flalign}\label{eq:convert}
&\lim_{n\rightarrow\infty}\frac{1}{n}\log\mP_{\theta,\sigma}\big(T(Q_n)\big) \nn\\&=-\inf_{\substack{\bm{U}=(U_1,...,U_K)\in (\mathcal{P}_{\mathcal{X}})^K\\\bm\alpha^T\bm{U} = Q}}\sum_{k=1}^K \alpha_k D(U_k||p_{\theta,k}).
\end{flalign}

The right hand side of \eqref{eq:convert} is a convex optimization problem with linear constraints, which can be solved efficiently using standard optimization tools \cite{cvx, cvxgb08}. Its computational complexity is independent of the number of sensors. Therefore, for large $n$, the mixture over $\sigma$ in \eqref{EQ:TYPE} can be approximated by solving a convex optimization problem whose computational complexity is independent of the network size $n$. 

Let $\bm{P} = [P_1\cdots P_K]^T$, where $P_k \in\mathcal{P}_\mathcal{X}$, for $1\leq k\leq K$.  For any $Q \in \mathcal{P}_\mathcal{X}$, define the following function of $Q$:
\begin{flalign}
	f_{\bm{P}}(\bm{\alpha},Q) = \inf\limits_{\substack{(U_1,...,U_K)\in(\mathcal{P}_{\mathcal{X}})^K\\ \bm\alpha^T\bm{U} = Q}}\sum_{k=1}^K \alpha_k D(U_k||P_k).
\end{flalign}

Intuitively, an algorithm for the problem in Section \ref{sec:2problemmodel} can be constructed by approximating the log of the mixture likelihood ratio at time $t$ in the mixture CuSum algorithm using $n\big(f_{\bm{P}_0}(\bm{\alpha},\Pi_{\mathbf X^n[t]})-f_{\bm{P}_1}(\bm{\alpha},\Pi_{\mathbf X^n[t]})\big)$.
However, the lower bound on the WARL for this algorithm is difficult to derive due to the ``inf" in the test statistic. We construct a novel test that can be updated recursively, and for which a lower bound on WARL can be theoretically derived. Moreover, as will be numerically demonstrated, this test has a WADD-WARL trade-off that is close to the optimal mixture CuSum, while also being computationally efficient.

Let $\hat{\nu}_t$ denote the change point estimate at time $t$. Denote by $\hat{t}\triangleq t-\hat{v}_t+1$. We then design our detection statistic to approximate $W[t]$ in \eqref{eq:mcusum}:
\begin{flalign}
	\widehat{W}[t] = \hat{t}n\big[ f_{\bm{P}_0}(\bm{\alpha},\Pi_{\mathbf X^n[\hat{\nu}_t,t]})-f_{\bm{P}_1}(\bm{\alpha},\Pi_{\mathbf X^n[\hat{\nu}_t,t]})\big].
\end{flalign}
Instead of using a maximum likelihood approach to estimate $\hat{\nu}_t$ as in \eqref{eq:mcusum}, which is not computationally efficient here, since $\hat{\nu}_t$ also appears in $\mathbf X^n[\hat{\nu}_t,t]$, we design a recursive way of updating $\hat{\nu}_t$. Let $\hat{\nu}_0 = 0$. If $\widehat{W}[t] \leq 0$, $\hat{\nu}_{t+1} = t+1$, and if $\widehat{W}[t] >0$, $\hat{\nu}_{t+1} = \hat{\nu}_t$. Then, $\Pi_{\mathbf X^n[\hat{\nu}_t,t]}$ can also be updated recursively: if $\widehat{W}[t] \leq 0$, $\Pi_{\mathbf X^n[\hat{\nu}_{t+1}, t+1]} = \Pi_{\mathbf X^n[t+1]}$, and if $\widehat{W}[t]>0$, $\Pi_{\mathbf X^n[\hat{\nu}_{t+1},t+1]} = \frac{\hat{t}\Pi_{\mathbf X^n[\hat{\nu}_t,t]} + \Pi_{\mathbf X^n[t+1]}}{\widehat{t}+1}$.

We next provide a heuristic explanation of how $\widehat{W}[t]$ evolves in the pre- and post-change regimes. According to the Glivenko–Cantelli theorem \cite{tucker1959}, before the change point $\nu$, as $n\rightarrow \infty$, $\Pi_{\mathbf X^n[\hat{\nu}_t, t]}$ converges to $\bm\alpha^T\bm P_0$ almost surely. It can be easily seen that $f_{\bm{P}}(\bm{\alpha},Q) \geq 0$ for any $\bm\alpha, \bm P$ and $Q$. The equality holds if and only if $\bm\alpha^T\bm P =  Q$. This implies that  $f_{\bm{P}_0}(\bm{\alpha},\bm\alpha^T\bm P_0)-f_{\bm{P}_1}(\bm{\alpha},\bm\alpha^T\bm P_0) < 0$. Therefore, before the change point $\nu$, for large $n$, $\widehat W[t]$ has a negative drift. Similarly, after the change point $\nu$, for large $n$, $\widehat W[t]$ has a positive drift of $f_{\bm{P}_0}(\bm{\alpha},\bm\alpha^T\bm P_1)$, and evolves towards $\infty$. This motivates us to construct the following computationally efficient test:
\begin{flalign}\label{eq:alter}
&\tau_e = \inf \Big\{t \geq 1: \widehat{W}[t] \geq b\Big\}.
\end{flalign}

The computation cost of $\tau_e$ mainly lies in the update of the empirical distribution and the optimization step. The computational complexity of updating the empirical distribution increases linearly with $n$, and the computational complexity of the optimization step is independent of $n$. Therefore, the computationally efficient test is more efficient than the optimal mixture CuSum algorithm when $n$ is large. Table \ref{table:com} summarizes the computational complexity of the mixture CuSum algorithm and the computationally efficient algorithm.

\begin{table}[!htbp]
\centering
\begin{tabular}{|c|c|c|}
\hline
  & Mixture CuSum & Efficient algorithm\\
\hline 
Complexity & $O(2^{nH(\bm\alpha)})$ & $O(n)$\\
\hline
\end{tabular}
\vspace{0.1cm}\caption{Computational complexity: mixture CuSum v.s. computationally efficient algorithm.}\label{table:com}
\end{table}

In the following theorem, we present a lower bound on the WARL for our computationally efficient test in \eqref{eq:alter}.
\begin{theorem}\label{THEOREM:1}
Define $\Gamma\triangleq \big\{\mu\in\mathcal{P}_{\mathcal{X}}: f_{\bm{P}_0}(\bm{\alpha},\mu) > f_{\bm{P}_1}(\bm{\alpha},\mu)\big\}$. Let
\begin{flalign}
h = \inf_{\substack{(U_1,...,U_K)\in(\mathcal{P}_{\mathcal{X}})^K\\ \bm\alpha^T\bm{U} \in \Gamma }}\sum_{k=1}^K n_k D(U_k||P_{0,k}).
\end{flalign}
Then $h>0$ and for any $\Omega$,
\begin{flalign}
\mE_\Omega^\infty \big[\tau_e(b)\big]\geq  \frac{e^{b}}{\big(\frac{b}{h}+1\big)\big(\prod_k|\mathcal{P}_{\frac{b}{h}n_k}|\big)}.
\end{flalign}
\end{theorem}
In the following, we provide a proof sketch, and the full proof can be found in Appendix \ref{sec: appc}.
\begin{proof}[Proof Sketch]
Let $Y = \text{inf}\{t \geq 1 : \widehat{W}[t] \leq 0\}$ be the first regeneration time.
For any $\Omega$ and $m\geq 1$, from Sanov's theorem \cite{cover2006information}, we can show that 
\begin{flalign}\label{eq:sketcheffarl}
\mP^\infty_\Omega(Y>m) \leq \bigg(\prod_k|\mathcal{P}_{mn_k}|\bigg) e^{-mh}.
\end{flalign}
	
Define regeneration times $Y_0 = 0$ and for $r\geq 0$, $Y_{r+1} = \inf\big\{t>Y_r : \widehat{W}[t] \leq 0\big\}$. 
Let 
$
R = \inf \{r: Y_r \leq \infty\ \text{and}\ \widehat{W}[t]\geq b\ \text{for some}\ Y_r < t\leq Y_{r+1}\}
$
denote the index of the first cycle in which $\widehat{W}[t]$ crosses $b$. Note that according to the recursive update rule of $\hat{\nu}_t$ and $\widehat{W}[t]$, the test statistics in cycle $r+1$ are independent of the samples in cycles $1,\cdots,r$.
For any $\Omega$, we have that
\begin{flalign}\label{eq:sketchex}
\mE_\Omega^\infty [\tau_e(b)] \geq \mE_\Omega^\infty [R] = \sum_{r=0}^\infty\mP_\Omega^\infty(R\geq r).
\end{flalign}

For any $\Omega$ and $m\geq 1$, we have that 
\begin{flalign}\label{eq:sketchstop}
&\mP_\Omega^\infty(\tau_e(b) < Y) \leq \mP_\Omega^\infty(\tau_e(b) < m) + \mP_\Omega^\infty(Y>m).
\end{flalign}
Consider the first term in \eqref{eq:sketchstop} $\mP_\Omega^\infty(\tau_e(b) < m)$, by applying Sanov's theorem \cite{cover2006information}, we have that for any $\Omega$, 
\begin{flalign}\label{eq:sketchfirst}
\mP_\Omega^\infty\big(\tau_e(b) < m\big) \leq m \left(\prod_k|\mathcal{P}_{mn_k}|\right)e^{-b}.
\end{flalign}
Let $m = \frac{b}{h}$. Combing \eqref{eq:sketcheffarl} and \eqref{eq:sketchfirst}, it follows that  
\begin{flalign}\label{eq:sketchrege}
&\mP_\Omega^\infty(\tau_e(b) < Y)
\leq \Big(\frac{b}{h} + 1\Big) \bigg(\prod_k|\mathcal{P}_{\frac{b}{h}n_k}|\bigg)e^{-b}. 
\end{flalign}

From \eqref{eq:sketchrege} and the independence among the cycles\cite{asmussen2008applied}, we have that 
\begin{flalign}\label{eq:sketchpro}
&\mP_\Omega^\infty(R\geq r) \geq \Bigg(1-\Big(\frac{b}{h}+1\Big)\bigg(\prod_k|\mathcal{P}_{\frac{b}{h}n_k}|\bigg)e^{-b}\Bigg)^r,
\end{flalign}
Therefore, from \eqref{eq:sketchex} and \eqref{eq:sketchpro}, for any $\Omega$,
\begin{flalign}
&\mE_\Omega^\infty [\tau_e(b)]\geq  \frac{e^{b}}{\Big(\frac{b}{h}+1\Big)\Big(\prod_k|\mathcal{P}_{\frac{b}{h}n_k}|\Big)}.
\end{flalign}
This completes the proof.
\end{proof}

To guarantee that $\inf_{\Omega} \mE_\Omega^\infty  \big[\tau_e(b)\big]\geq \gamma$, it suffices to choose $b$ such that $\frac{e^{b}}{\big(\frac{b}{h}+1\big)\big(\prod_k|\mathcal{P}_{\frac{b}{h}n_k}|\big)} = \gamma$ and $b\sim\log\gamma$. 

Note that an upper bound on the WADD for $\tau_e$ is difficult to obtain. To understand the detection delay of the proposed computationally efficient test, we then study the case when the change occurs at $\nu=1$. We have the following result.
\begin{proposition}\label{proposition:1}
Consider the case with $\nu=1$.  Then, as $t\rightarrow\infty$,
	\begin{flalign}
	&n\big[ f_{\bm{P}_0}(\bm{\alpha},\Pi_{\mathbf X^n[1,t]})-f_{\bm{P}_1}(\bm{\alpha},\Pi_{\mathbf X^n[1,t]})\big]\rightarrow nf_{\bm{P}_0}(\bm{\alpha},\bm{\alpha}^T\bm{P}_1),\nn
	\end{flalign} 
	almost surely.
\end{proposition}
\begin{proof}
	According to the Glivenko–Cantelli theorem \cite{tucker1959}, as $t\rightarrow \infty$, under the post-change distribution, the empirical distribution $\Pi_{\mathbf X^n[1,t]}$ convergences to $\bm{\alpha}^T\bm{P}_1$ almost surely. Due to the fact that $f_{\bm{P}_1}(\bm{\alpha},\bm{\alpha}^T\bm{P}_1) = 0$, we have that $\frac{1}{t}(t-1+1)n\big[ f_{\bm{P}_0}(\bm{\alpha},\Pi_{\mathbf X^n[1,t]})-f_{\bm{P}_1}(\bm{\alpha},\Pi_{\mathbf X^n[1,t]})\big]$ converges to $nf_{\bm{P}_0}(\bm{\alpha},\bm{\alpha}^T\bm{P}_1)$ almost surely.
\end{proof}
Intuitively, Proposition \ref{proposition:1} implies that if the change is at $\nu=1$ and regeneration does not happen, then the detection delay of the computationally efficient algorithm increases linearly with the threshold $b$ at the rate of ${1}/(nf_{\bm{P}_0}(\bm{\alpha},\bm{\alpha}^T\bm{P}_1))$.

We then present the following universal lower bound on the WADD, and show that the slope is also ${1}/(nf_{\bm{P}_0}(\bm{\alpha},\bm{\alpha}^T\bm{P}_1))$ when $n$ is large.
\begin{proposition}\label{proposition:2}
	For large $\gamma$, we have that 
	\begin{flalign}
	\inf_{\tau: \text{WARL}\geq \gamma} \text{WADD}(\tau) \sim \frac{\log \gamma}{D\Big(\widetilde{\mP}_1\big|\big|\widetilde{\mP}_0\Big)}(1+o(1)). 	
	\end{flalign}
	Moreover, as $n\rightarrow\infty$,
	\begin{flalign}\label{EQ:KLCONVERGE}
	\lim_{n\rightarrow\infty} \frac{1}{n}D\Big(\widetilde{\mP}_1\big|\big|\widetilde{\mP}_0\Big)= f_{\bm{P}_0}(\bm{\alpha},\bm{\alpha}^T\bm{P}_1).
	\end{flalign}
\end{proposition}
\begin{proof}[Proof Sketch]
	It was shown in Section \ref{sec:exopmixture} that the mixture CuSum $\tau^*$ is exactly optimal for the QCD problem in Section \ref{sec:modelqcd}. Then, as $\gamma\rightarrow \infty$, we have that $\inf_{\tau: \text{WARL}\geq \gamma} \text{WADD}(\tau) = \text{WADD}(\tau^*)$.
	Further note that for the mixture CuSum $\tau^*$, $\tau^*$ achieves the equality in \eqref{EQ:AAA}. Then, we have that $\widetilde{\text{WADD}}(\tau^*)= {\text{WADD}}(\tau^*)$. Since $\tau^*$ is optimal for the simple QCD problem in \eqref{eq:newproblem}, from Theorem 4 in \cite{lai1998information},  as $\gamma\rightarrow \infty$, it follows that 
	\begin{flalign}
	\text{WADD}(\tau^*) = \widetilde{\text{WADD}}(\tau^*) \sim \frac{\log\gamma}{D\Big(\widetilde{\mP}_1\big|\big|\widetilde{\mP}_0\Big)}(1+o(1)).
	\end{flalign}
	The proof of \eqref{EQ:KLCONVERGE} can be found in Appendix \ref{sec:KL}.	
\end{proof}
By combining Propositions \ref{proposition:1} and \ref{proposition:2},  it can be seen that the tradeoff between the WADD and WARL for our computationally efficient test is close to the optimal one when $n$ is large. This demonstrates the advantage of our test that for large networks, it has a similar statistical efficiency comparing to the optimal test, and has a significantly reduced  computational complexity.

\section{Simulation Results}\label{sec:7numerical}
\subsection{Mixture CuSum Algorithm}
We first show an example evolution path of the mixture CuSum algorithm. Set $n=2$ and $K=2$, i.e., one sensor in each group. For group 1, the pre- and post-change distributions are $\mathcal N(0,1)$ and $\mathcal N(0.5,1)$, respectively. For group two, the pre- and post-change distribution are $\mathcal N(2,1)$ and $\mathcal N(1.5,1)$, respectively. 
In Fig.~\ref{fig:1}, we set the change point to be 500 and $b=5$. We plot one sample evolution path of the mixture CuSum algorithm. It can be seen that before the change point, the test statistic fluctuates around zero, and after the change point, it starts to increase with a positive drift.


We then compare our optimal mixture CuSum test with two other heuristic algorithms based on the Bayesian approach and the generalized likelihood ratio approach to tackling the unknown group assignments. For the  Bayesian approach, we pretend  that each sample comes from group $k$ with probability $n_k/n$, for $k=1,\ldots,K$, independently, so that {on average} the $k$-th group has $n_k$ sensors, although we actually have exact $n_k$ sensors in each group $k$.  We then compute the following likelihood ratio: 
\begin{flalign}\label{eq:bayes}
    l_b(x^n[t])  = \frac{\prod_{i=1}^n\left(\sum_{k=1}^K\frac{n_k}{n}p_{1,k}(x_i[t])\right)}{\prod_{i=1}^n\left(\sum_{k=1}^K\frac{n_k}{n}p_{0,k}(x_i[t])\right)}.    
\end{flalign}
The generalized likelihood ratio for the sample $x^n[t]$ is 
\begin{equation}\label{eq:glrt}
    l_g(x^n[t]) = \frac{\sup_{\sigma\in \ms}\mathbb{P}_{1,\sigma} (x^n[t])}{\sup_{\sigma\in \ms}\mathbb{P}_{0,\sigma} (x^n[t])}.
\end{equation}
We then design CuSum-type tests using \eqref{eq:bayes} and \eqref{eq:glrt}, which are referred to as Bayesian and Generalized CuSums. The test statistics of these three algorithms are all symmetric, and therefore for different $\Omega$, the average detection delay and average run length are the same. 

In Fig.~\ref{fig:2}, we plot the WADD as a function of the WARL. It can be seen  that our mixture CuSum algorithm outperforms the other two algorithms. Moreover, the relationship between the WADD and log of the WARL is linear. The slope of these three curves should be the reciprocal  of the expectation of the corresponding likelihood ratio under $\mP_{1,\sigma}$ for some $\sigma\in\ms$. Due to the fact that the distributions are continuous,  our computationally efficient test is not applicable here.

\begin{figure*}[!t]
	\begin{multicols}{3}
		\includegraphics[width=\linewidth]{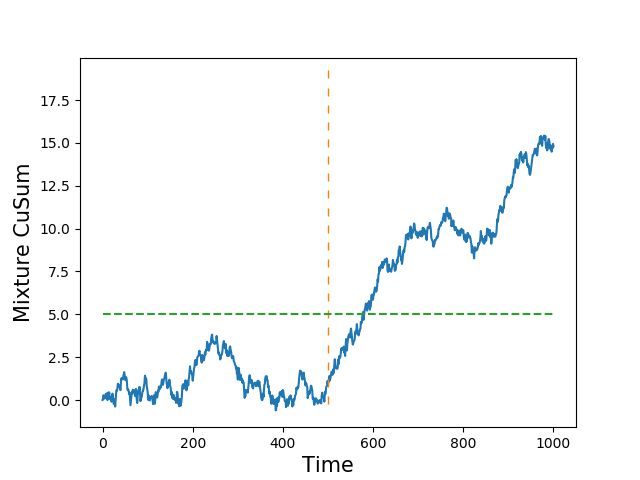}\par\caption{Evolution path of the Mixture CuSum.}\label{fig:1}
		\includegraphics[width=\linewidth]{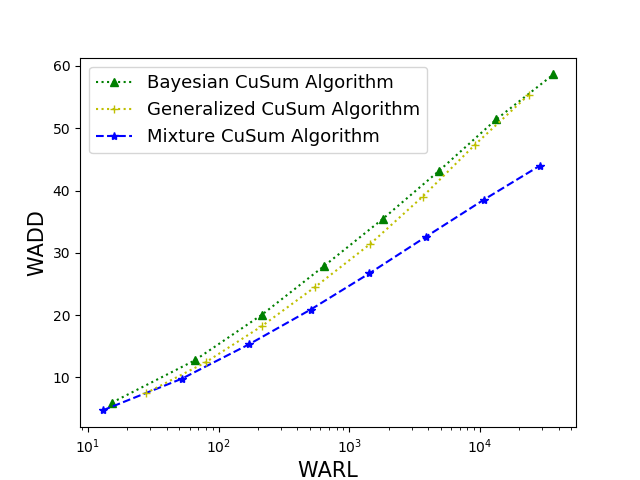}\par\caption{Comparison of the Bayesian CuSum algorithm, the Generalized CuSum algorithm and the Mixture CuSum algorithm: $n=2, K=2$}\label{fig:2}
		\includegraphics[width=\linewidth]{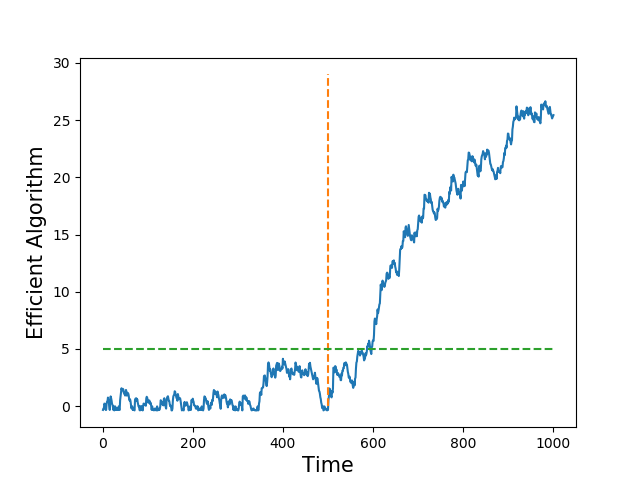}\par\caption{Evolution path of the computationally efficient algorithm: $n=2, K=2$.}\label{fig:6}
	\end{multicols}
	\begin{multicols}{3}
		\includegraphics[width=\linewidth]{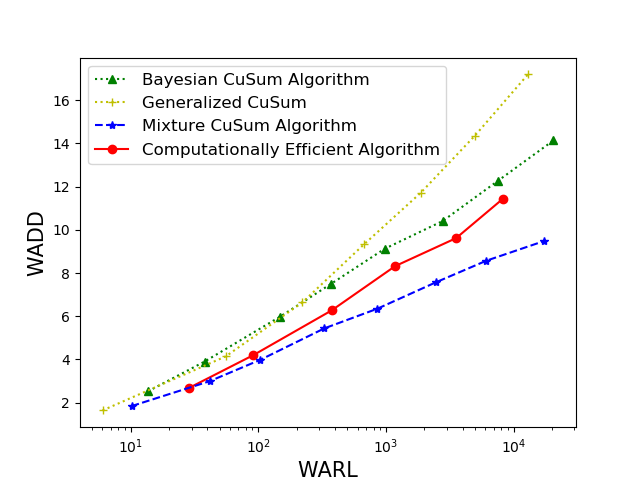}\par\caption{Comparison of the Bayesian CuSum algorithm, the Generalized CuSum algorithm, the Mixture CuSum algorithm and the computationally efficient algorithm: $n=2, K=2$.}\label{fig:3}
		\includegraphics[width=\linewidth]{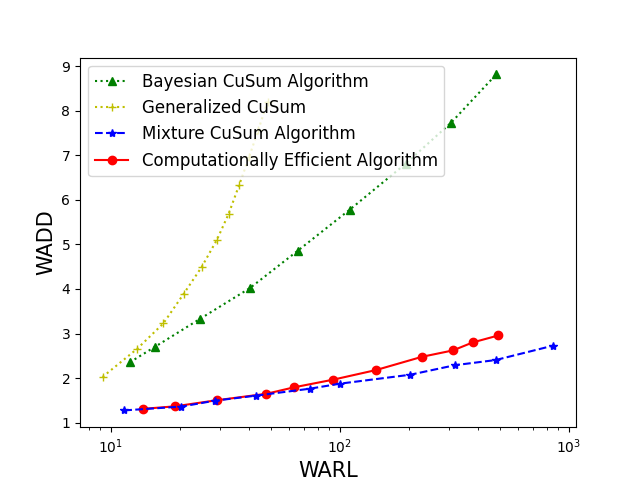}\par\caption{Comparison of the Bayesian CuSum algorithm, the Generalized CuSum algorithm, the Mixture CuSum algorithm and the computationally efficient algorithm: $n=8, K=2$.}\label{fig:4}
		\includegraphics[width=\linewidth]{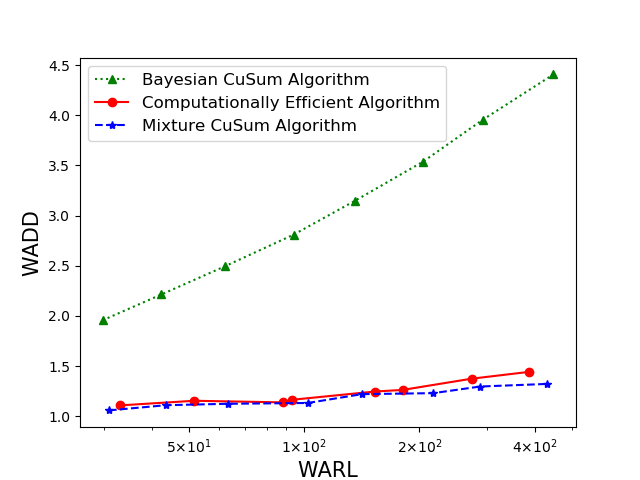}\par\caption{Comparison of the Bayesian CuSum algorithm, the Mixture CuSum algorithm and the computationally efficient algorithm: $n=20, K=2$.}\label{fig:5}
	\end{multicols}
    \begin{multicols}{3}
    	\includegraphics[width=\linewidth]{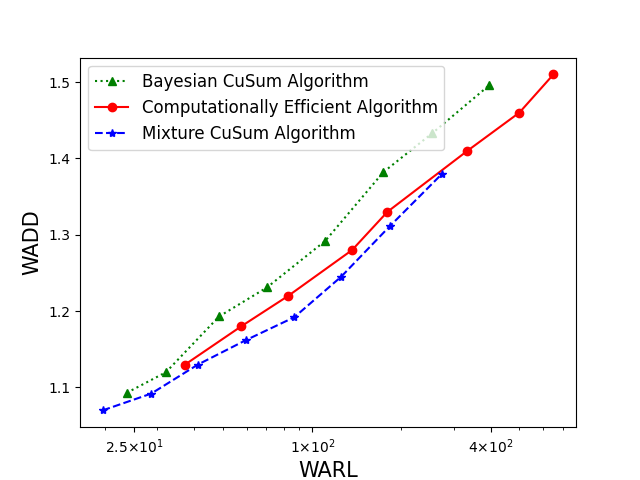}\par\caption{Comparison of the Bayesian CuSum algorithm, the Mixture CuSum algorithm and the computationally efficient algorithm: $n=10, K=4$.}\label{fig:7}
    	\includegraphics[width=\linewidth]{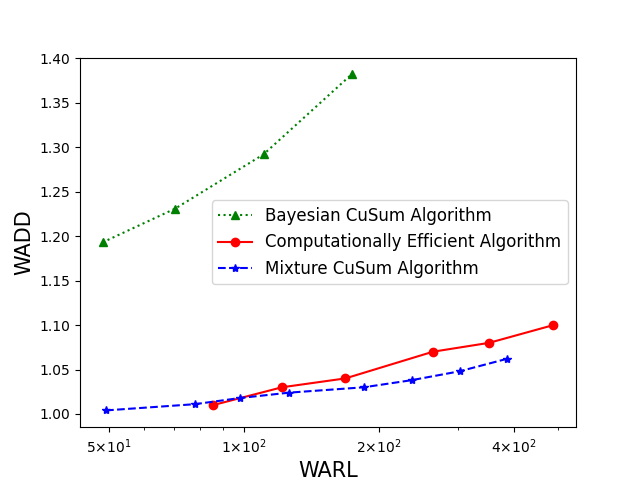}\par\caption{Comparison of the Bayesian CuSum algorithm, the Mixture CuSum algorithm and the computationally efficient algorithm: $n=18, K=4$.}\label{fig:8}
    	\includegraphics[width=\linewidth]{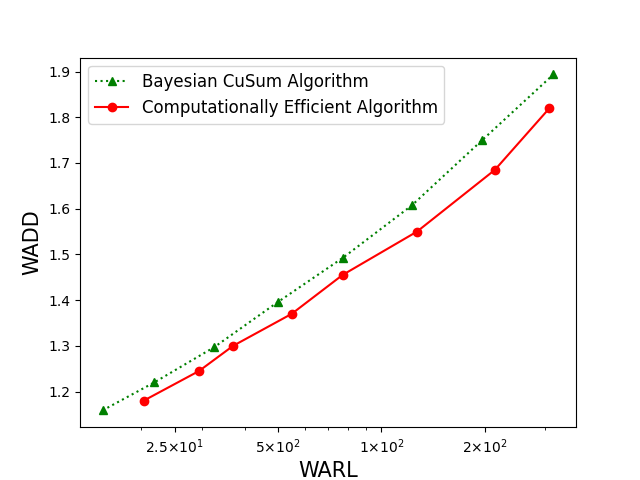}\par\caption{Comparison of the Bayesian CuSum algorithm and the computationally efficient algorithm: $n=100, K=4$.}\label{fig:10}
    \end{multicols}
    
\end{figure*}

\subsection{Computationally Efficient Algorithm}
For the computationally efficient algorithm, we first consider a simple example with $n=2$, $K=2$, $n_1=1$ and $n_2=1$. The pre- and post-change distributions for group 1 are binomial distribution $\mathcal B(10,0.5)$ and $\mathcal B(10,0.3)$, respectively, and for group 2 are $\mathcal B(10,0.5)$ and $\mathcal B(10,0.7)$, respectively. We plot a sample evolution path of the efficient algorithm in Fig.~\ref{fig:6}. Similar to the  mixture CuSum, before the change point, the test statistic fluctuates around zero, and after the change point, it starts to increase with a positive drift. 

We then compare the performance of our efficient algorithm with the optimal mixture CuSum algorithm, the Bayesian CuSum algorithm and the Generalized CuSum algorithm, and repeat the experiment for $n=8$, $n_1=4$, $n_2=4$ and $n=20$, $n_1 = 10, n_2 = 10$ with the same distributions.

For the three cases with $n=2$, $n=8$ and $n=20$, we plot the WADD as a function of the WARL in Figs.~\ref{fig:3}, Fig.~\ref{fig:4} and Fig.~\ref{fig:5}. It can be seen that mixture CuSum outperforms the other three tests, and our computationally efficient test has a better performance than the intuitive Bayesian CuSum and Generalized CuSum. For the case with $n=20$, $n_1 = 10, n_2 = 10$, the performance of the Generalized CuSum algorithm is much worse than the other three algorithms,  therefore
 is not included in Fig.~\ref{fig:5}. More importantly, comparing Fig.~\ref{fig:3}, Fig.~\ref{fig:4} and Fig.~\ref{fig:5}, we can see that as $n$ increases, the slope of the WADD-WARL tradeoff curve of the efficient algorithm gets closer to the one of the optimal mixture CuSum algorithm.
This conforms to the design of our computationally efficient test which aims to approximate the optimal mixture CuSum when $n$ is large, and our theoretical discussion in Propositions \ref{proposition:1} and \ref{proposition:2}.

We then consider the case with $K=4$. The pre- and post-change distributions for group 1 are $\mathcal B(10,0.5)$ and $\mathcal B(10,0.3)$, respectively, for group 2 are $\mathcal B(10,0.5)$ and $\mathcal B(10,0.7)$, respectively, for group 3 are $\mathcal B(10,0.5)$ and $\mathcal B(10,0.25)$, respectively, for group 4 are $\mathcal B(10,0.5)$ and $\mathcal B(10,0.75)$, respectively. In Fig.~\ref{fig:7}, we plot the WADD as a function of the WARL with $n = 10, n_1 = n_2 = n_3 = 2$ and $n_4 = 4$. In Fig.~\ref{fig:8}, we plot the WADD as a function of the WARL with $n = 18, n_1 = n_2 = n_3 = 3$ and $n_4 = 9$. From Fig.~\ref{fig:7} and Fig.~\ref{fig:8}, it can be seen that the mixture CuSum has the best performance, and our computationally efficient algorithm outperforms other heuristic algorithms, and is close to the optimal mixture CuSum algorithm. We also compare the performance of Bayesian CuSum algorithm and the efficient algorithm when $n=100$ and $K=4$. We note that for the case $n=100$, the computational complexity is too high for the mixture CuSum algorithm to be simulated. We set $n_1 = n_2 = n_3 = 20$ and $n_4 = 40$, and the pre- and post-change distributions for group 1 are $\mathcal B(10,0.5)$ and $\mathcal B(10,0.4)$, respectively, for group 2 are $\mathcal B(10,0.5)$ and $\mathcal B(10,0.45)$, respectively, for group 3 are $\mathcal B(10,0.5)$ and $\mathcal B(10,0.35)$, respectively, for group 4 are $\mathcal B(10,0.5)$ and $\mathcal B(10,0.6)$, respectively. From Fig.~\ref{fig:10}, it can be seen that the performance of our efficient algorithm is better than the Bayesian CuSum algorithm.

In Fig.~\ref{fig:9}, we show the computational efficiency of our proposed algorithms. Specifically, we compare the running time of computing one step update of our computationally efficient algorithm and the optimal mixture CuSum algorithm (on Intel Core i5-8265U CPU). From Fig.~\ref{fig:3}, one can see that as $n$ increases, the running time of the mixture CuSum increases exponentially, while the running time of our computationally efficient test  stays almost the same. 

\begin{figure}[htbp]
	\includegraphics[width=3.2in]{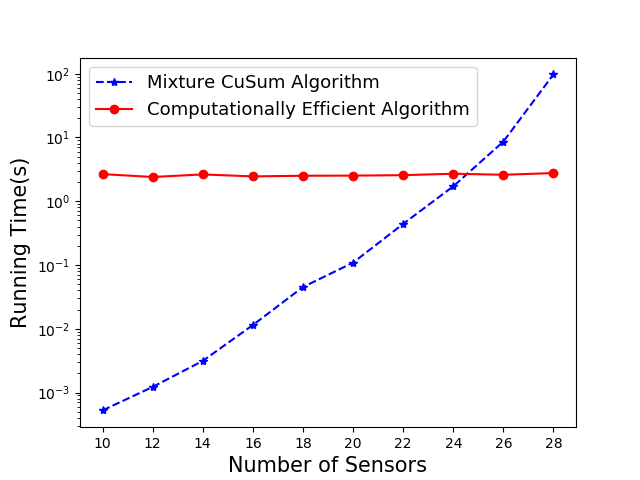}\par\caption{Comparison of the computational complexity between the Mixture CuSum algorithm and the computationally efficient algorithm.}\label{fig:9}
\end{figure}

\section{Conclusion}\label{sec:8conclusion}
In this paper, we studied the statistical inference problem in anonymous heterogeneous sensor networks. We first revisited the  non-sequential setting studied in \cite{anonymous}, and provided a simple optimality proof for the MLRT. We then extended our approach to the problem of QCD with anonymous heterogeneous sensors, and constructed a mixture CuSum algorithm. We showed that the mixture CuSum algorithm is optimal under Lorden's criterion \cite{lorden1971procedures}. We note that asymptotic optimality results can also be obtained under  Pollak's criterion \cite{pollak1985}. Although being optimal, our mixture CuSum algorithm is computationally expensive when the number of sensors is large. We then proposed a computationally efficient algorithm with a novel recursive update rule of the change point estimate and the test statistic. We further developed its WARL lower bound for practical false alarm control. Our numerical results showed that the mixture CuSum algorithm has the best performance and the computationally efficient algorithm also outperforms other heuristic algorithms. Moreover, when the number of sensor is large, the computationally efficient algorithm is much more efficient than the optimal mixture CuSum algorithm. Our results provide useful tools and insights to investigate various kinds of statistical inference problems in anonymous networks.


One possible extension  is to the case where the samples in different time steps are not independent \cite{lai1998information}. It is also of interest to investigate when samples are quantized and sensors can only receive binary codewords \cite{cheng2029multi}. In this case, such quantizing measurement should be incorporated into the design of mixture CuSum algorithm. Moreover, in this paper, it is assumed that after the change all the sensors change their data-generating distributions simultaneously. Therefore, another possible future direction is to consider the case where only an unknown subset of sensors are affected by the change.  Moreover, the change may also be dynamic and propagate following some unknown pattern.  In many practical applications, the data-generating distributions may not be available beforehand, and data-driven approaches in anonymous heterogeneous networks need to be developed.

\appendices

\section{Proof of \eqref{EQ:AAA}}\label{proff}
We construct a new sequence of random variables $\{\widehat{X}^n[t]\}_{t=1}^\infty$. Before the change point, $\widehat{X}^n[t]$ are i.i.d. according to the mixture distribution $\widetilde{\mP}_0=\frac{1}{\mid \ms\mid}\sum_{\sigma\in \ms}\mP_{0,\sigma}$. After the change point, $\widehat{X}^n[t]$ follows the distribution $\mP_{1,\sigma_t}$ for some $\sigma_t\in\ms$. Specifically,
\begin{flalign}\label{eq:new}
\widehat{X}^n[t] \sim \left\{\begin{array}{ll}
\widetilde{\mP}_0, &\text{ if } t<\nu \\
\mP_{1,\sigma_t}, &\text{ if } t\geq \nu.
\end{array}\right.
\end{flalign}

For any stopping time $\tau$, define the worst-case average detection delay for the model in \eqref{eq:new} as follows:
\begin{flalign}
&\widehat{\text{WADD}}(\tau)\nn\\&=\sup\limits_{\nu\geq1}\sup_{\sigma_\nu,...,\sigma_\infty}\text{ess}\sup{\widehat{\mE}}^\nu_{\sigma_\nu,...,\sigma_\infty}[(\tau-\nu)^+|\widehat{\mathbf X}^n[1,\nu-1]],\nn
\end{flalign}
where $\widehat{\mE}^\nu_{\sigma_\nu,...,\sigma_\infty}$ denotes the expectation when the data is distributed according to \eqref{eq:new}.
To prove that $\text{WADD}(\tau)\geq \widetilde{\text{WADD}}(\tau)$, we will first show that ${\text{WADD}}(\tau) =\widehat{\text{WADD}}(\tau)$, and then show that $\widehat{\text{WADD}}(\tau) \geq \widetilde{\text{WADD}}(\tau)$. 

\textbf{Step 1.} Denote by $\mathcal{M}$ the collection of all $\{\sigma_1,...,\sigma_{\nu-1}\}$, and $\mu$ is an element in $\mathcal{M}$. Denote by $\mathcal{N}$ the collection of all $\{\sigma_\nu,...,\sigma_\infty\}$, and $\omega$ is an element in $\mathcal{N}$. Thus, $\Omega = \{\mu,\omega\}$. Then, the $\text{WADD}$ can be written as
\begin{flalign}
&\text{WADD}(\tau)\nn \\&= \sup\limits_{\nu\geq1}\sup_{\Omega}\text{ess}\sup{\mE}^\nu_{\Omega}[(\tau-\nu)^+|{\mathbf X}^n[1,\nu-1]]\nn \\&= \sup\limits_{\nu\geq1}\sup_{\omega\in\mathcal{N}}\sup\limits_{\mu\in\mathcal{M}}\text{ess}\sup{\mE}^\nu_{\omega}[(\tau-\nu)^+|{\mathbf X}^n[1,\nu-1]],
\end{flalign}
where $\mE^\nu_\omega$ denotes the expectation when change point is $\nu$, and after the change point, the data follows distribution $\prod_{t=\nu}^\infty\mP_{1,\sigma_t}$.
We note that $\widehat{X}^n[t]$ and $X^n[t]$, for $t\geq\nu$, have the same distribution $\mP_{1,\sigma_t}$. Therefore, the difference between $\text{WADD}$ and $\widehat{\text{WADD}}$ lies in that they take esssup with respect to different distributions, i.e., the distributions of $\mathbf{X}^n[1,\nu-1]$ and $\mathbf{\widehat{X}}^n[1,\nu-1]$ are different. Let $f_{\omega}(\mathbf {X}^n[1,\nu-1])$ denote ${\mE}^\nu_{\omega}[(\tau-\nu)^+|\mathbf{X}^n[1,\nu-1]]$.
Then, $\text{WADD}$ and $\widehat{\text{WADD}}$ can be written as
\begin{flalign}
\text{WADD}(\tau) &= \sup\limits_{\nu\geq1}\sup_{\omega\in\mathcal{N}}\sup_{\mu\in\mathcal{M}}\text{ess}\sup f_{\omega}({\mathbf X}^n[1,\nu-1]),\nn\\\widehat{\text{WADD}}(\tau) &= \sup\limits_{\nu\geq1}\sup_{\omega\in\mathcal{N}}\text{ess}\sup f_{\omega}(\widehat{\mathbf X}^n[1,\nu-1]).
\end{flalign}
It then suffices to show that 
$
\sup_{\mu\in\mathcal{M}}\text{ess}\sup f_{\omega}({\mathbf X}^n[1,\nu-1])=\text{ess}\sup f_{\omega}(\widehat{\mathbf X}^n[1,\nu-1]). 
$

For any $\omega \in \mathcal{N}$ and $\mu \in \mathcal{M}$, let
$
b_{\omega,\mu} = \text{ess}\sup f_\omega({\mathbf X}^n[1,\nu-1])=\inf\{b: \mP_{\mu}(f_{\omega}({\mathbf X}^n[1,\nu-1])>b)=0\},
$
where $\mP_\mu$ denotes the probability measure when the data is generated according to $\mP_{0,\sigma_1},...,\mP_{0,\sigma_{\nu-1}}$ before change point $\nu$.

Let $b^*_\omega = \text{ess}\sup f_\omega(\widehat{\mathbf{X}}^n[1,\nu-1])$. It can be shown that
\begin{flalign*}
b^*_\omega &=
\inf\bigg\{b:\int_{\textbf{x}^n[1,\nu-1]} \mathbbm{1}_{\{f_\omega(\textbf{x}^n[1,\nu-1])>b\}}\nn\\&\hspace{1cm}\times\mathrm{d} \prod_{t=1}^{\nu-1}\widetilde{\mP}_0(x^n(t))=0\bigg\}
\nn\\&=\inf\bigg\{b:\int_{\textbf{x}^n[1,\nu-1]} \mathbbm{1}_{\{f_\omega(\textbf{x}^n[1,\nu-1])>b\}}\nn\\&\hspace{1cm}\times\mathrm{d} \prod_{t=1}^{\nu-1}\frac{1}{\mid \ms \mid}\sum_{\sigma_t \in \ms}{\mP}_{0,\sigma_t}(x^n(t))=0\bigg\}
\nn\\&=\inf\bigg\{b:\int_{\textbf{x}^n[1,\nu-1]} \mathbbm{1}_{\{f_\omega(\textbf{x}^n[1,\nu-1])>b\}}\nn\\&\hspace{1cm}\times\mathrm{d}\frac{1}{\mid \mathcal{M}\mid}\sum_{\mu\in \mathcal{M}}{\mP}_{\mu}(\textbf{x}^n[1,\nu-1])=0\bigg\}
\nn\\&=\inf\bigg\{b: \frac{1}{\mid \mathcal{M}\mid}\sum_{\mu\in \mathcal{M}}{\mP}_{\mu}(f_{\omega}({{\mathbf X}}^n[1,\nu-1])>b)=0\bigg\}.
\end{flalign*}
It then follows that for any $\mu\in\mathcal{M}$,
\begin{flalign}
{\mP}_{\mu}(f_{\omega}({{\mathbf X}}^n[1,\nu-1])>b^*_\omega) = 0.
\end{flalign}
Therefore, for any $\mu\in\mathcal{M}$, we have that $b_{\omega,\mu}\leq b^*_\omega$. Then
\begin{flalign}\label{eq:achieve}
\sup\limits_{\mu\in\mathcal{M}} b_{\omega,\mu} \leq b^*_\omega.
\end{flalign}

Conversely, let $\sup_{\mu\in\mathcal{M}} b_{\omega,\mu}=b'$. For any $\mu\in\mathcal{M}$, we have $
\mP_{\mu} (f_{\omega}({\mathbf X}^n[1,\nu-1])>b')= 0.
$
Then,
$
\frac{1}{\mid \mathcal{M}\mid}\sum_{\mu \in \mathcal{M}}\mP_{\mu}(f_{\omega}({\mathbf X}^n[1,\nu-1])>b') = 0.
$
This further implies that 
\begin{flalign}\label{eq:converse}
b^*_\omega\leq b'=\sup\limits_{\mu\in\mathcal{M}} b_{\omega,\mu}.
\end{flalign}

Combining \eqref{eq:achieve} and \eqref{eq:converse}, we have that 
$
\sup_{\mu\in\mathcal{M}} b_{\omega,\mu}=b^*_\omega,
$
and thus
$
\sup_{\mu\in\mathcal{M}}\text{ess}\sup f_{\omega}({\mathbf X}^n[1,\nu-1])= \text{ess}\sup f_{\omega}(\widehat{\mathbf X}^n[1,\nu-1]).
$
This implies that 
\begin{flalign}\label{eq:first}
\text{WADD}(\tau) = \widehat{\text{WADD}}(\tau).
\end{flalign}

\textbf{Step 2.} The next step is to show that $\widehat{\text{WADD}}(\tau) \geq \widetilde{\text{WADD}}(\tau)$. We will first show that
$
\sup_{\omega\in\mathcal{N}}{\text{ess}\sup f_{\omega}(\widehat{\mathbf X}^n[1,\nu-1])}\geq {\text{ess}\sup\sup_{\omega\in\mathcal{N}} f_{\omega}(\widehat{\mathbf X}^n[1,\nu-1])}.
$
Denote by $\widetilde{\mP}^\nu$ the probability measure when the change is at $\nu$, the pre- and post-change distributions are $\widetilde{\mP}_0$ and $\widetilde{\mP}_1$, respectively. Let $\hat{b}=\sup_{\omega\in\mathcal{N}}{\text{ess}\sup f_{\omega}(\widehat{\mathbf X}^n[1,\nu-1])}$.
For any $\omega\in\mathcal{N}$, we have that
$
\widetilde{\mP}^\nu(f_{\omega}(\widehat{\mathbf X}^n[1,\nu-1])\geq \hat{b}) = 0.
$
Since $\mathcal{N}$ is countable, it then follows that 
\begin{flalign}
&\widetilde{\mP}^\nu\Big(\sup\limits_{\omega\in\mathcal{N}} f_{\omega}(\widehat{\mathbf X}^n[1,\nu-1])\geq \hat{b}\Big)\nn\\& \leq \widetilde{\mP}^\nu\Big(\cup_{\omega\in\mathcal{N}}\big\{f_\omega(\widehat{\mathbf X}^n[1,\nu-1])>\hat{b}\big\}\Big)\nn\\&\leq\sum_{\omega\in\mathcal{N}}\widetilde{\mP}^\nu\Big(f_\omega(\widehat{\mathbf X}^n[1,\nu-1])>\hat{b}\Big) = 0.
\end{flalign}
Therefore,
\begin{flalign}\label{eq:supp}
\hat{b}&=\sup\limits_{\omega\in\mathcal{N}} \text{ess}\sup f_{\omega}(\widehat{\mathbf X}^n[1,\nu-1]) \nn\\&\geq \text{ess}\sup\sup\limits_{\omega\in\mathcal{N}} f_{\omega}(\widehat{\mathbf X}^n[1,\nu-1]).
\end{flalign}

Before the change point $\nu$, $\widehat{X}^n[t]$ and $\widetilde{X}^n[t]$ follow the same distribution. For any $T\geq\nu+1$, we have that
\begin{flalign}
&\sup_{\substack{\{\sigma_\nu,...,\sigma_T\}\\\in\ms^{\bigotimes(T-\nu+1)}}}\sum_{t=\nu+1}^T (t-\nu)\mP_{\sigma_\nu,...,\sigma_T}(\tau=t|\widehat{\mathbf{X}}^n[1,\nu-1])\nn\\
&\geq \sum_{t=\nu+1}^T (t-\nu)\frac{1}{\mid \ms \mid^{(T-\nu+1)}}\nn\\
&\hspace{0.5cm}\times\sum\limits_{\substack{\{\sigma_\nu,...,\sigma_T\}\\\in\ms^{\bigotimes(T-\nu+1)}}}\mP_{\sigma_\nu,...,\sigma_T}(\tau=t|\widehat{\mathbf{X}}^n[1,\nu-1])\nn\\&=\sum_{t=\nu+1}^T (t-\nu)\widetilde{\mP}^\nu(\tau=t|\widetilde{\mathbf{X}}^n[1,\nu-1]).
\end{flalign}
As $T\rightarrow\infty$, we have that
\begin{flalign}\label{eq:ave}
\widehat{\mE}^\nu_{\omega}[(\tau-\nu)|\widehat{\mathbf X}^n[1,\nu-1]]\geq
\widetilde{\mE}^\nu[(\tau-\nu)|\widetilde{\mathbf X}^n[1,\nu-1]],
\end{flalign}
where $\mP_{\sigma_\nu,...,\sigma_T}$ denotes the probability measure when the observations from time $\nu$ to time $T$ are generated according to ${\mP}_{\sigma_\nu},...,{\mP}_{\sigma_T}$.

From \eqref{eq:supp} and \eqref{eq:ave}, we have that
\begin{flalign}\label{eq:second}
&\widehat{\text{WADD}}(\tau) \nn\\&= \sup\limits_{\omega\in\mathcal{N}} \text{ess}\sup f_{\omega}(\widehat{\mathbf X}^n[1,\nu-1])\nn\\&\geq \text{ess}\sup\widetilde{\mE}^\nu[(\tau-\nu)^+|\widetilde{\mathbf X}^n[1,\nu-1]]\nn\\&=
\widetilde{\text{WADD}}(\tau).
\end{flalign}

Combining \eqref{eq:first} and \eqref{eq:second}, it follows that $\text{WADD}(\tau) = \widehat{\text{WADD}}(\tau) \geq \widetilde{\text{WADD}}(\tau)$. Similarly, it can be shown that $\text{WARL}(\tau) \leq \widetilde{\text{ARL}}(\tau)$. This concludes the proof.

\section{$\tau^*$ achieves equality in \eqref{EQ:AAA}}\label{sec:appachiev}
We will show that the mixture CuSum algorithm achieves the equality in \eqref{EQ:AAA}, i.e.,
\begin{flalign}\label{eq:opt}
\widehat{\text{WADD}}(\tau^*) = \widetilde{\text{WADD}}(\tau^*).
\end{flalign}

For any $\{\sigma_\nu,...,\sigma_i,...,\sigma_\infty\}$, consider another element in $\mathcal{N}$, $\{\sigma_\nu,...,\sigma'_i,...\sigma_\infty\}$. Due to the fact that $\tau^*$ is symmetric, it follows that for any $i\geq\nu$, and any $\sigma_i,\sigma'_i\in\ms$, 
\begin{flalign}\label{eq:ttt}
&\text{ess}\sup\widehat{\mE}^\nu_{\sigma_\nu,...,\sigma_i,...,\sigma_\infty}[(\tau^*-\nu)^+|\widehat{\mathbf X}^n[1,\nu-1]]\nn \\&= \text{ess}\sup\widehat{\mE}^\nu_{\sigma_\nu,...,\sigma'_i,...,\sigma_\infty}[(\tau^*-\nu)^+|\widehat{\mathbf X}^n[1,\nu-1]].
\end{flalign}
Therefore,
$\widehat{\text{WADD}}(\tau^*)$ doesn't depend on $\omega$, which further implies that
\begin{flalign}\label{eq:sym}
&\sup_{\omega\in\mathcal{N}}\text{ess}\sup\widehat{\mE}^\nu_{\omega}[(\tau^*-\nu)^+|\widehat{\mathbf X}^n[1,\nu-1]]\nn\\&=\text{ess}\sup\widehat{\mE}^\nu_{\omega}[(\tau^*-\nu)^+|\widehat{\mathbf X}^n[1,\nu-1]].
\end{flalign}

For any $T\geq\nu+1$, we have that
\begin{flalign}
&\sup_{\substack{\{\sigma_\nu,...,\sigma_T\}\\\in\ms^{\bigotimes(T-\nu+1)}}}\sum_{t=\nu+1}^T (t-\nu)\mP_{\sigma_\nu,...,\sigma_T}(\tau^*=t|\widehat{\mathbf{X}}^n[1,\nu-1])\nn\\&
= \sum_{t=\nu+1}^T (t-\nu)\frac{1}{\mid \ms \mid^{(T-\nu+1)}}\nn\\&\quad\times\sum\limits_{\substack{\{\sigma_\nu,...,\sigma_T\}\\\in\ms^{\bigotimes(T-\nu+1)}}}\mP_{\sigma_\nu,...,\sigma_T}(\tau^*=t|\widehat{\mathbf{X}}^n[1,\nu-1])\nn\\&
=\sum_{t=\nu+1}^T (t-\nu)\widetilde{\mP}^\nu(\tau^*=t|\widetilde{\mathbf{X}}^n[1,\nu-1]).
\end{flalign}
As $T\rightarrow\infty$, we have that
\begin{flalign}\label{eq:equal}
&\widehat{\mE}^\nu_{\omega}[(\tau^*-\nu)^+|\widehat{\mathbf X}^n[1,\nu-1]]= \widetilde{\mE}^\nu[(\tau^*-\nu)^+|\widetilde{\mathbf X}^n[1,\nu-1]]. 
\end{flalign}
From \eqref{eq:sym} and \eqref{eq:equal}, it follows that
\begin{flalign}
&\widehat{\text{WADD}}(\tau^*) \nn\\&= \sup_{\nu\geq 1}\sup_{\omega\in\mathcal{N}}\text{ess}\sup\widehat{\mE}^\nu_{\omega}[(\tau^*-\nu)^+|\widehat{\mathbf X}^n[1,\nu-1]]\nn\\&=\sup_{\nu\geq1} \text{ess}\sup\widehat{\mE}^\nu_{\omega}[(\tau^*-\nu)^+|\widetilde{\mathbf X}^n[1,\nu-1]]\nn\\&= \sup_{\nu\geq1}\text{ess}\sup\widetilde{\mE}^\nu[(\tau^*-\nu)^+|\widetilde{\mathbf X}^n[1,\nu-1]]\nn\\&= \widetilde{\text{WADD}}(\tau^*).
\end{flalign}
Similarly, it can be shown that $\widetilde{\text{ARL}}(\tau^*) = \text{WARL}(\tau^*)$.

\section{Proof of Theorem \ref{THEOREM:1}}\label{sec: appc}
Let $Y = \text{inf}\{t \geq 1 : \widehat{W}[t] \leq 0\}$ be the first regeneration time.
For any $\Omega$ and $m\geq 1$, we have that 
\begin{flalign}
&\mP^\infty_\Omega(Y>m)= \mP^\infty_\Omega\big(\widehat{W}[t]>0, \forall t \in [1,m]\big)\nn\\&\leq \mP_\Omega^\infty\Big(nm\big[f_{\bm{P}_0}(\bm{\alpha},\Pi_{\mathbf X^n[1,m]})-f_{\bm{P}_1}(\bm{\alpha},\Pi_{\mathbf X^n[1,m]})\big]>0\Big)\nn.
\end{flalign}
Let $\Gamma \triangleq \{\mu\in\mathcal{P}_{\mathcal{X}}|f_{\bm{P}_0}(\bm{\alpha},\mu) > f_{\bm{P}_1}(\bm{\alpha},\mu)\}$. We have that 
\allowdisplaybreaks
\begin{align}\label{eq:expbound}
&\mP_\Omega^\infty\Big(nm\big[f_{\bm{P}_0}(\bm{\alpha},\Pi_{\mathbf X^n[1,m]})-f_{\bm{P}_1}(\bm{\alpha},\Pi_{\mathbf X^n[1,m]})\big]>0\Big) \nn\\& = \mP_\Omega^\infty \big\{\Pi_{\mathbf X^n[1,m]} \in \Gamma\big\}\nn\\& = \sum_{\substack{(U_1,...,U_K)\in\mathcal{P}_{mn_1}\times ... \times \mathcal{P}_{mn_K}\\ \bm\alpha^T\bm{U} \in \Gamma}}\prod_{k=1}^K p_{0,k}^{\bigotimes mn_k}\big(T_{mn_k}(U_k)\big)\nn\\& \leq \sum_{\substack{(U_1,...,U_K)\in\mathcal{P}_{mn_1}\times ... \times \mathcal{P}_{mn_K}\\ \bm\alpha^T\bm{U} \in \Gamma}} e^{-\sum_{k=1}^K mn_kD(U_k||p_{0,k})}\nn\\&\leq \bigg(\prod_k|\mathcal{P}_{mn_k}|\bigg)\nn\\&\hspace{0.2cm}\cdot\exp\Bigg(-\hspace{-0.2cm}\inf\limits_{\substack{(U_1,...,U_K)\in\mathcal{P}_{mn_1}\times ... \times \mathcal{P}_{mn_K}\\ \bm\alpha^T\bm{U} \in \Gamma }}\sum_{k=1}^K mn_k D(U_k||p_{0,k})\Bigg)\nn\\&\leq \bigg(\prod_k|\mathcal{P}_{mn_k}|\bigg)e^{-hm},
\end{align}
where the last step is due to the fact that $\mathcal{P}_{mn_k}\subseteq \mathcal{P}_\mathcal{X}$, $\forall 1\leq k\leq K$.
Note that $f_{\bm{P}}(\bm{\alpha}, Q) \geq 0$ for any $Q$ and the equality holds if and only if $\bm\alpha^T\bm P = Q$ almost everywhere. We then have that $\bm\alpha^T\bm P_0 \notin\Gamma$ and $h>0$. Therefore, for any $\Omega$ and $m\geq 1$,
\begin{flalign}\label{eq:effarlpart1}
\mP^\infty_\Omega(Y>m) \leq \bigg(\prod_k|\mathcal{P}_{mn_k}|\bigg) e^{-mh}.
\end{flalign}

Define regeneration times $Y_0 = 0$ and for $r\geq 0$, $Y_{r+1} = \inf\big\{t>Y_r : \widehat{W}[t] \leq 0\big\}$. 
Let 
$
R = \inf \{r: Y_r \leq \infty\ \text{and}\ \widehat{W}[t]\geq b\ \text{for some}\ Y_r < t\leq Y_{r+1}\}
$
denote the index of the first cycle in which $\widehat{W}[t]$ crosses $b$. Note that according to the recursive update rule of $\hat{\nu}_t$ and $\widehat{W}[t]$, the test statistics in cycle $r+1$ are independent of the samples in cycles $1,\cdots,r$.
For any $\Omega$, we have that
\begin{flalign}
\mE_\Omega^\infty [\tau_e(b)] \geq \mE_\Omega^\infty [R] = \sum_{r=0}^\infty\mP_\Omega^\infty(R\geq r).
\end{flalign}
	
For any $\Omega$ and $m\geq 1$, we have that 
\begin{flalign}\label{eq:stopfirst}
&\mP_\Omega^\infty(\tau_e(b) < Y)\nn\\& = \mP_\Omega^\infty(\tau_e(b) < Y, Y\leq m) + \mP_\Omega^\infty(\tau_e(b) < Y, Y > m) \nn\\&\leq \mP_\Omega^\infty(\tau_e(b) < m) + \mP_\Omega^\infty(Y>m).
\end{flalign}
	
Consider the first term in \eqref{eq:stopfirst} $\mP_\Omega^\infty(\tau_e(b) < m)$:
\begin{flalign*}
&\mP_\Omega^\infty\big(\tau_e(b) < m\big) 
= \mP_\Omega^\infty\Big(\max_{1\leq t< m} \widehat{W}[t] \geq b\Big) \nn\\& \leq \sum_{1\leq t <m}\mP_\Omega^\infty\Big(\widehat{W}[t]\geq b\Big)\nn\\
& = \hspace{-0.3cm}\sum_{1\leq t <m}\hspace{-0.2cm}\mP_\Omega^\infty\Big(n\hat{t}\big[f_{\bm{P}_0}(\bm{\alpha},\Pi_{\mathbf X^n[\hat{\nu}_t, t]})-f_{\bm{P}_1}(\bm{\alpha},\Pi_{\mathbf X^n[\hat{\nu}_t, t]})\big]\hspace{-0.1cm}\geq b\Big).
\end{flalign*}
Define $\Gamma_{b,t}\triangleq \Big\{\mu\in\mathcal{P}_{\mathcal{X}} \big| n\hat{t}\big[f_{\bm{P}_0}(\bm{\alpha},\mu) - f_{\bm{P}_1}(\bm{\alpha},\mu)\big]\geq b\Big\}$. For all $\mu \in \Gamma_{b,t}$, we have that 
$
n\hat{t}f_{\bm{P}_0}(\bm{\alpha},\mu) \geq b + n\hat{t}f_{\bm{P}_1}(\bm{\alpha},\mu) \geq b,
$
where the last inequality is due to the facts that $\hat{t}\geq 0$ and $f_{\bm{P}_1}(\bm{\alpha},\mu) \geq 0$. 
For any $\Omega$ and $1\leq t<m$, following the same idea as the one in \eqref{eq:expbound}, we have that 
\begin{flalign}\label{eq:effarlpart2}
&\mP_\Omega^\infty\Big(n\hat{t}\big[f_{\bm{P}_0}(\bm{\alpha},\Pi_{\mathbf X^n[\hat{\nu}_t, t]})-f_{\bm{P}_1}(\bm{\alpha},\Pi_{\mathbf X^n[\hat{\nu}_t, t]})\big]>b\Big)\nn\\& =\mP_\Omega^\infty\Big\{\Pi_{\mathbf X^n[\hat{\nu}_t, t]}\in \Gamma_{b,t}\Big\}\nn\\& 
\leq\bigg(\prod_k|\mathcal{P}_{\hat{t}n_k}|\bigg)\nn\\&\hspace{0.2cm}\cdot\exp\Bigg(-\hspace{-0.2cm}\inf\limits_{\substack{(U_1,...,U_K)\in(\mathcal{P}_{\mathcal{X}})^K\\ \bm\alpha^T\bm{U} \in \Gamma_{b,t} }}\sum_{k=1}^K n_k\hat{t} D(U_k||p_{0,k})\Bigg)\nn\\&
\leq \bigg(\prod_k|\mathcal{P}_{mn_k}|\bigg)e^{-b}.
\end{flalign}
We then have that for any $\Omega$, 
\begin{flalign}
\mP_\Omega^\infty\big(\tau_e(b) < m\big) \leq m \left(\prod_k|\mathcal{P}_{mn_k}|\right)e^{-b}.
\end{flalign}
	
Let $m = \frac{b}{h}$. Combing \eqref{eq:effarlpart1} and \eqref{eq:effarlpart2}, we have that  
\begin{flalign}\label{eq:probound}
&\mP_\Omega^\infty(\tau_e(b) < Y)
\leq \Big(\frac{b}{h} + 1\Big) \bigg(\prod_k|\mathcal{P}_{\frac{b}{h}n_k}|\bigg)e^{-b}. 
\end{flalign}
	
It then follows that 
\begin{flalign}
&\mP_\Omega^\infty(R\geq r) = \mP_\Omega^\infty\Big(\widehat{W}[t]<b,\forall t\leq Y_r\Big)\nn\\&= \mP_\Omega^\infty\Big(\widehat{W}[t]<b,\forall Y_{m-1}\leq t\leq Y_m, \forall 1\leq m\leq r\Big)\nn\\&=\prod_{m=1}^r \mP_\Omega^\infty\Big(\widehat{W}[t]<b,\forall Y_{m-1}\leq t\leq Y_m\Big)\nn\\&\geq \Bigg(1-\Big(\frac{b}{h}+1\Big)\bigg(\prod_k|\mathcal{P}_{\frac{b}{h}n_k}|\bigg)e^{-b}\Bigg)^r,
\end{flalign}
where the second equality is due to \eqref{eq:probound} and the independence among the cycles\cite{asmussen2008applied}.
Therefore, for any $\Omega$,
\begin{flalign}
&\mE_\Omega^\infty [\tau_e(b)]\geq \sum_{r=0}^\infty \Bigg(1-\Big(\frac{b}{h}+1\Big)\bigg(\prod_k|\mathcal{P}_{\frac{b}{h}n_k}|\bigg)e^{-b}\Bigg)^r \nn\\& = \frac{e^{b}}{\Big(\frac{b}{h}+1\Big)\Big(\prod_k|\mathcal{P}_{\frac{b}{h}n_k}|\Big)}.
\end{flalign}
This completes the proof.

\section{Proof of \eqref{EQ:KLCONVERGE}}\label{sec:KL}
From \eqref{EQ:TYPE}, we have that for any $\sigma\in\ms$,
\begin{flalign}
&\log\frac{\widetilde{\mP}_1(X^n)}{\widetilde{\mP}_0(X^n)} =\log\frac{\mP_{1,\sigma}\big(T(\Pi_{X^n})\big)}{\mP_{0,\sigma}\big(T(\Pi_{X^n})\big)}.
\end{flalign}
Let $\mathcal{B}(\bm{\alpha}^T\bm{P}_\theta,\epsilon) = \big\{\mu\in\mathcal{P}_{\mathcal{X}}: \sup\limits_{x\in\mathcal{X}} \big|\mu(x)-\bm{\alpha}^T\bm{P}_\theta(x)\big| \leq \epsilon\big\}$ denote the ball centered at $\bm{\alpha}^T\bm{P}_\theta$ with radius $\epsilon > 0$. According to the Glivenko–Cantelli theorem \cite{tucker1959}, we then have that for any $\sigma\in\ms$ and $\epsilon > 0$,
\begin{flalign}
\lim_{n\rightarrow \infty}\mP_{\theta, \sigma}\Big\{\sup\limits_{x\in\mathcal{X}} \big|\Pi_{X^n}(x)-\bm{\alpha}^T\bm{P}_\theta(x)\big| > \epsilon\Big\} = 0.
\end{flalign} 
It then follows that for any $\sigma\in\ms$ and $\epsilon >0$,
\begin{flalign}\label{eq:typeconverge}
&\lim_{n\rightarrow \infty}\mP_{\theta, \sigma}\Big\{\Pi_{X^n}\notin \mathcal{B}(\bm{\alpha}^T\bm{P}_\theta,\epsilon) \Big\}\nn\\& = \lim_{n\rightarrow \infty}\mP_{\theta, \sigma}\Big\{\sup\limits_{x\in\mathcal{X}} \big|\Pi_{X^n}(x)-\bm{\alpha}^T\bm{P}_\theta(x)\big| > \epsilon\Big\} = 0.
\end{flalign}
It was shown in Lemma 5.3 in \cite{anonymous} that  $f_{\bm{P}_\theta}(\bm{\alpha},P)$ is a continuous function of $P$ for any $\theta \in \{0,1\}$. Therefore,  $f_{\bm{P}_0}(\bm{\alpha},P)-f_{\bm{P}_1}(\bm{\alpha},P)$ is a continuous function of $P$.
Then we have that for any $\epsilon>0$, there exists an $\eta(\epsilon) >0$ such that $\forall P \in \mathcal{B}(\bm{\alpha}^T\bm{P}_1,\epsilon)$, 
\begin{flalign}\label{eq:continuous}
&f_{\bm{P}_0}(\bm{\alpha},\bm{\alpha}^T\bm{P}_1)-\eta(\epsilon) <f_{\bm{P}_0}(\bm{\alpha}, P)-f_{\bm{P}_1}(\bm{\alpha},P)\nn\\& < f_{\bm{P}_0}(\bm{\alpha},\bm{\alpha}^T\bm{P}_1)+\eta(\epsilon),
\end{flalign}
where $\eta(\epsilon) \rightarrow 0$ as $\epsilon\rightarrow 0$.
We then have that 
\begin{flalign}\label{eq:exponent}
&\lim_{n\rightarrow\infty}\frac{1}{n}D\Big(\widetilde{\mP}_1\big|\big|\widetilde{\mP}_0\Big) \nn\\& = \lim_{n\rightarrow\infty}\frac{1}{n} \mE_{\widetilde{\mP}_1}\Big[\log\mP_{1,\sigma}\big(T(\Pi_{X^n})\big) - \log\mP_{0,\sigma}\big(T(\Pi_{X^n})\big) \Big]\nn\\& \overset{(a)}{\leq} \lim_{n\rightarrow\infty}\frac{1}{n} \mE_{\widetilde{\mP}_1}\Bigg[\log \bigg(\prod_k|\mathcal{P}_{n_k}|\bigg) - \log\bigg(\prod_{k=1}^K\frac{1}{(n_k+1)^{\mid \mathcal{X}\mid}}\bigg)\nn\\&\hspace{0.5cm} - \inf_{\substack{(U_1,...,U_K)\in\mathcal{P}_{n_1}\times ... \times \mathcal{P}_{n_K}\\ \bm\alpha^T\bm{U} = \Pi_{X^n}}}\sum_{k=1}^K n_k D(U_k||p_{1,k})\nn\\& \hspace{0.5cm} +\inf_{\substack{(U_1,...,U_K)\in\mathcal{P}_{n_1}\times ... \times \mathcal{P}_{n_K}\\ \bm\alpha^T\bm{U} = \Pi_{X^n}}}\sum_{k=1}^K n_kD(U_k||p_{0,k}) \Bigg]\nn\\&= \lim_{n\rightarrow\infty}\frac{1}{n} \widetilde{\mP}_1\big(\Pi_{X^n}\in \mathcal{B}(\bm{\alpha}^T\bm{P}_1,\epsilon)\big)\mE_{\widetilde{\mP}_1}\Bigg[\log \bigg(\prod_k|\mathcal{P}_{n_k}|\bigg) \nn\\&\hspace{0.5cm} - \inf_{\substack{(U_1,...,U_K)\in\mathcal{P}_{n_1}\times ... \times \mathcal{P}_{n_K}\\ \bm\alpha^T\bm{U} = \Pi_{X^n}}}\sum_{k=1}^K n_k D(U_k||p_{1,k})\nn\\& \hspace{0.5cm} +\inf_{\substack{(U_1,...,U_K)\in\mathcal{P}_{n_1}\times ... \times \mathcal{P}_{n_K}\\ \bm\alpha^T\bm{U} = \Pi_{X^n}}}\sum_{k=1}^K n_kD(U_k||p_{0,k})\nn\\&\hspace{0.5cm}- \log\bigg(\prod_{k=1}^K\frac{1}{(n_k+1)^{\mid \mathcal{X}\mid}}\bigg)\Bigg|\Pi_{X^n}\in \mathcal{B}(\bm{\alpha}^T\bm{P}_1,\epsilon)\Bigg]\nn\\&\hspace{0.5cm}+\lim_{n\rightarrow\infty}\frac{1}{n} \widetilde{\mP}_1\big(\Pi_{X^n}\notin \mathcal{B}(\bm{\alpha}^T\bm{P}_1,\epsilon)\big)\mE_{\widetilde{\mP}_1}\Bigg[\log \bigg(\prod_k|\mathcal{P}_{n_k}|\bigg) \nn\\&\hspace{0.5cm} - \inf_{\substack{(U_1,...,U_K)\in\mathcal{P}_{n_1}\times ... \times \mathcal{P}_{n_K}\\ \bm\alpha^T\bm{U} = \Pi_{X^n}}}\sum_{k=1}^K n_k D(U_k||p_{1,k})\nn\\& \hspace{0.5cm} +\inf_{\substack{(U_1,...,U_K)\in\mathcal{P}_{n_1}\times ... \times \mathcal{P}_{n_K}\\ \bm\alpha^T\bm{U} = \Pi_{X^n}}}\sum_{k=1}^K n_kD(U_k||p_{0,k})\nn\\&\hspace{0.5cm}- \log\bigg(\prod_{k=1}^K\frac{1}{(n_k+1)^{\mid \mathcal{X}\mid}}\bigg)\Bigg|\Pi_{X^n}\notin \mathcal{B}(\bm{\alpha}^T\bm{P}_1,\epsilon)\Bigg]\nn\\& \overset{(b)}{=} \lim_{n\rightarrow\infty}\frac{1}{n} \widetilde{\mP}_1\big(\Pi_{X^n}\in \mathcal{B}(\bm{\alpha}^T\bm{P}_1,\epsilon)\big)\mE_{\widetilde{\mP}_1}\Bigg[ \nn\\&\hspace{0.5cm} - \inf_{\substack{(U_1,...,U_K)\in\mathcal{P}_{n_1}\times ... \times \mathcal{P}_{n_K}\\ \bm\alpha^T\bm{U} = \Pi_{X^n}}}\sum_{k=1}^K n_k D(U_k||p_{1,k})\nn\\& \hspace{0.5cm} +\inf_{\substack{(U_1,...,U_K)\in\mathcal{P}_{n_1}\times ... \times \mathcal{P}_{n_K}\\ \bm\alpha^T\bm{U} = \Pi_{X^n}}}\sum_{k=1}^K n_kD(U_k||p_{0,k})\nn\\&\hspace{0.5cm}\Bigg|\Pi_{X^n}\in \mathcal{B}(\bm{\alpha}^T\bm{P}_1,\epsilon)\Bigg]\nn\\&\hspace{0.5cm} + \lim_{n\rightarrow\infty}\frac{1}{n} \widetilde{\mP}_1\big(\Pi_{X^n}\in \mathcal{B}(\bm{\alpha}^T\bm{P}_1,\epsilon)\big)\bigg(\log \Big(\prod_k|\mathcal{P}_{n_k}|\Big)\nn\\&\hspace{0.5cm}- \log\Big(\prod_{k=1}^K\frac{1}{(n_k+1)^{\mid \mathcal{X}\mid}}\Big)\bigg)\nn\\&\overset{(c)}{\leq} f_{\bm{P}_0}(\bm{\alpha},\bm{\alpha}^T\bm{P}_1) + \eta(\epsilon),
\end{flalign}
where the inequality (a) is due to the bound of the probability of type classes\cite{cover2006information}: $\frac{1}{(n_k+1)^{\mid \mathcal{X}\mid}}2^{-n_kD(U_k||p_{\theta,k})}\leq p_{\theta;k}^{\bigotimes n_k}\big(T_{n_k}(U_k)\big)\leq 2^{-n_kD(U_k||p_{\theta,k})}$, the equality (b) is due to the fact that $\lim_{n\rightarrow\infty}\widetilde{\mP}_1\big(\Pi_{X^n}\notin \mathcal{B}(\bm{\alpha}^T\bm{P}_1,\epsilon)\big) =0$ and the inequality (c) is due to \eqref{eq:continuous} and the fact that $\lim_{n\rightarrow\infty}\frac{1}{n} \widetilde{\mP}_1\big(\Pi_{X^n}\in \mathcal{B}(\bm{\alpha}^T\bm{P}_1,\epsilon)\big)\bigg(\log \Big(\prod_k|\mathcal{P}_{n_k}|\Big)- \log\Big(\prod_{k=1}^K\frac{1}{(n_k+1)^{\mid \mathcal{X}\mid}}\Big)\bigg) = 0$.

For the lower bound, following the same idea as in \eqref{eq:exponent}, we have that 
\begin{flalign}\label{eq:exponent1}
&\lim_{n\rightarrow\infty}\frac{1}{n}D\Big(\widetilde{\mP}_1\big|\big|\widetilde{\mP}_0\Big) \nn\\& \geq \lim_{n\rightarrow\infty}\frac{1}{n} \mE_{\widetilde{\mP}_1}\Bigg[\log\bigg(\prod_{k=1}^K\frac{1}{(n_k+1)^{\mid \mathcal{X}\mid}}\bigg)- \log \bigg(\prod_k|\mathcal{P}_{n_k}|\bigg)\nn\\&\hspace{0.5cm} - \inf_{\substack{(U_1,...,U_K)\in\mathcal{P}_{n_1}\times ... \times \mathcal{P}_{n_K}\\ \bm\alpha^T\bm{U} = \Pi_{X^n}}}\sum_{k=1}^K n_k D(U_k||p_{1,k})\nn\\& \hspace{0.5cm} +\inf_{\substack{(U_1,...,U_K)\in\mathcal{P}_{n_1}\times ... \times \mathcal{P}_{n_K}\\ \bm\alpha^T\bm{U} = \Pi_{X^n}}}\sum_{k=1}^K n_kD(U_k||p_{0,k}) \Bigg]\nn\\&\geq f_{\bm{P}_0}(\bm{\alpha},\bm{\alpha}^T\bm{P}_1) - \eta(\epsilon).
\end{flalign}
By \eqref{eq:exponent} and \eqref{eq:exponent1}, we have that 
\begin{flalign}
\lim_{n\rightarrow\infty}\frac{1}{n}D\Big(\widetilde{\mP}_1\big|\big|\widetilde{\mP}_0\Big) = f_{\bm{P}_0}(\bm{\alpha},\bm{\alpha}^T\bm{P}_1).
\end{flalign}

\normalem
\bibliographystyle{ieeetr}
\bibliography{QCD2}

\begin{IEEEbiography}[{\includegraphics[width=1in,height=1.25in,clip,keepaspectratio]{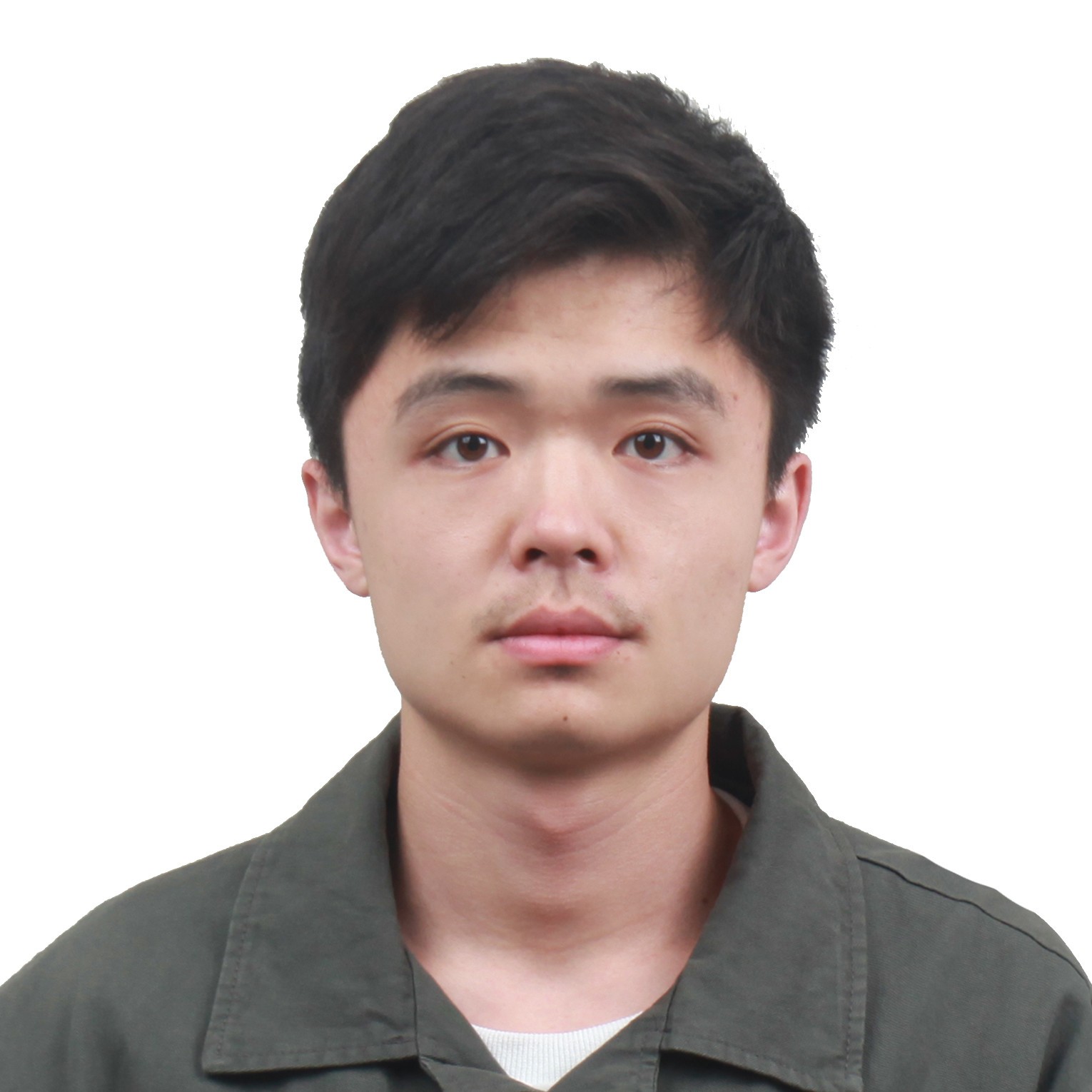}}]{Zhongchang Sun}(S'20) is a PhD student at the Department of Electrical Engineering, University at Buffalo, the State University of New York. He received the B.S. degree from Beijing Institute of Technology, Beijing, China in 2019. His research interests are on hypothesis testing, quickest change detection and distributionally robust optimization.
	
\end{IEEEbiography}

\begin{IEEEbiography}[{\includegraphics[width=1in,height=1.25in,clip,keepaspectratio]{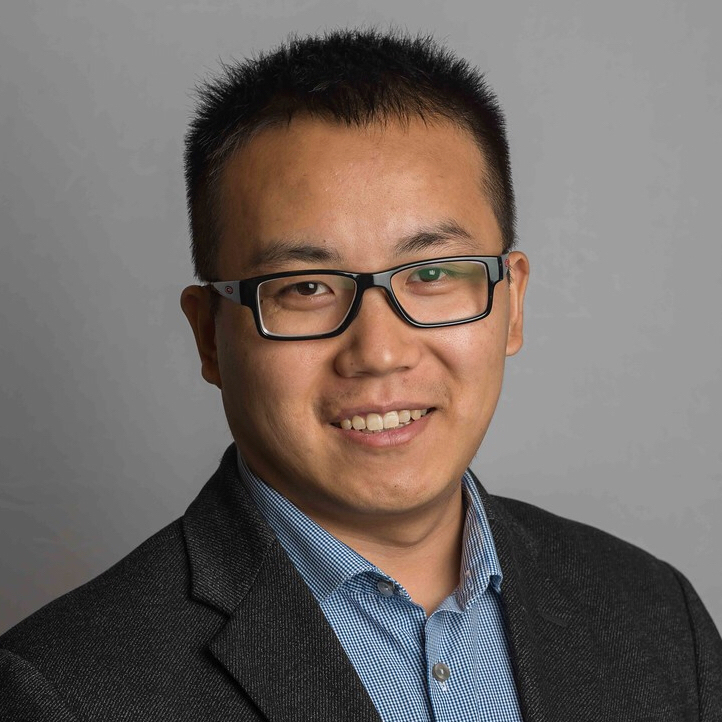}}]{Shaofeng Zou}(S'14-M'16) is an Assistant Professor, at the Department of Electrical Engineering, University at Buffalo, the State University of New York. He received the Ph.D. degree in Electrical and Computer Engineering from Syracuse University in 2016. He received the B.E. degree (with honors) from Shanghai Jiao Tong University, Shanghai, China, in 2011. He was a postdoctoral research associate at the Coordinated Science Lab, University of Illinois at Urbana-Champaign during 2016-2018. Dr. Zou's research interests include reinforcement learning, machine learning, statistical signal processing and information theory. He received the National Science Foundation CRII award in 2019. 
	
\end{IEEEbiography}

\begin{IEEEbiography}[{\includegraphics[width=1in,height=1.25in,clip,keepaspectratio]{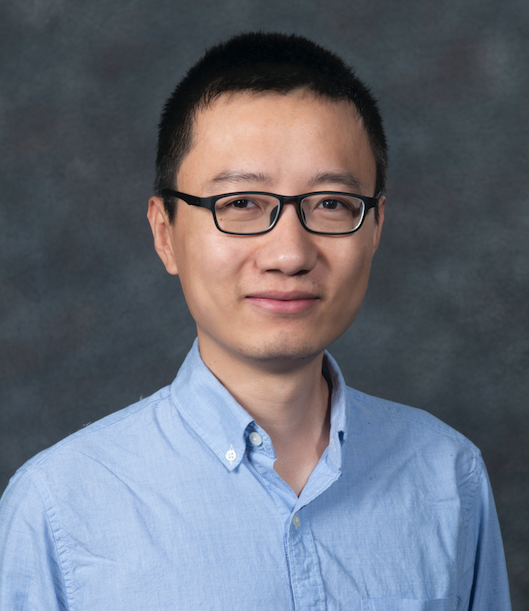}}]{Ruizhi Zhang} is an Assistant Professor in the Department of Statistics at University of Nebraska-Lincoln. He received his B.S. degree in Mathematics from Hua Loo-Keng Talent Program in Mathematics at University of Science and Technology of China (USTC) in 2014, graduated with honors. He received his Ph.D. degree in Statistics in the School of Industrial and Systems Engineering at Georgia Institute of Technology. His research interests include change-point detection, sequential analysis, robust statistics, high-dimensional statistical inference, functional data analysis.
\end{IEEEbiography}

\begin{IEEEbiography}[{\includegraphics[width=1in,height=1.25in,clip,keepaspectratio]{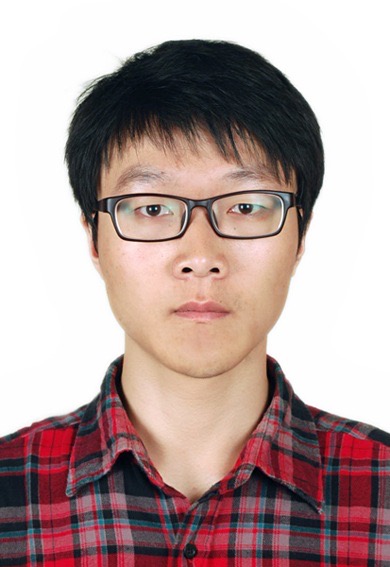}}]{Qunwei~Li}
(S'16) received the B.S.\ and M.S.\ degrees in electrical engineering with honors from Xidian University, Xi'an, China, in 2011 and 2014. He received the Ph.D. degree in electrical engineering from Syracuse University, Syracuse, NY, USA, in 2018.
He was a postdoctoral researcher in the Center for Applied Scientific Computing (CASC) at Lawrence Livermore National Laboratory (LLNL), CA, USA, from 2018 to 2019. His research interests include human decision making, adversarial deep learning, optimization algorithms, and recommender systems.
Dr. Li received the Syracuse University Graduate Fellowship Award in 2014 and the All University Doctoral Prize 2018 by Syracuse University for superior achievement in completed dissertations.
\end{IEEEbiography}

\end{document}